\documentclass[a4paper, 12pt]{article}
\usepackage[english]{babel}
\usepackage{amsfonts,amssymb,epsfig}
\usepackage{color,graphicx,graphics,psfrag}
\usepackage{amsmath,amstext,amscd}
\usepackage{hyperref}
\usepackage{amsthm, cases}
\usepackage[percent]{overpic}
\usepackage[titletoc, title]{appendix}
\usepackage{enumitem}
\usepackage{empheq}
\usepackage{caption}
\usepackage{subcaption}
\usepackage{ bbold }

 \allowdisplaybreaks[3]

\def\Box{\leavevmode\vbox{\hrule
     \hbox{\vrule\kern4pt\vbox{\kern4pt}%
           \vrule}\hrule}}

\def\paragraph#1{{\bf #1\ }}

\numberwithin{equation}{section}
\newtheorem{lemma}{Lemma}[section]
\newtheorem{theorem}[lemma]{Theorem}

\newtheorem{definition}[lemma]{Definition}
\newtheorem{proposition}[lemma]{Proposition}
\newtheorem{remark}{Remark}[section]

\newcommand{\lp}{\left(}
\newcommand{\rp}{\right)}
\newcommand{\eps}{\varepsilon}
\newcommand{\R}{\mathbb{R}}
\newcommand{\Id}{\mbox{Id}}

\newcommand{\beqar}{\begin{eqnarray*}}
\newcommand{\eeqar}{\end{eqnarray*}}

\newcommand{\beqarl}{\begin{eqnarray}}
\newcommand{\eeqarl}{\end{eqnarray}}

\newcommand{\be}{\begin{equation}}
\newcommand{\ee}{\end{equation}}

\title{Coupled Self-Organized Hydrodynamics and Stokes models for suspensions of active particles }
\author{Pierre Degond$^1$, Sara Merino-Aceituno$^1$, Fabien Vergnet$^2$ and Hui Yu$^3$ }
\date{\vspace{-5ex}}
\begin{document}
\maketitle
\medskip
{\footnotesize
 \centerline{1. Department of Mathematics, Imperial College London}
   \centerline{London, SW7 2AZ, United Kingdom}
 \centerline{pdegond@imperial.ac.uk, s.merino-aceituno@imperial.ac.uk}
}

\medskip
{\footnotesize
 \centerline{2. Laboratoire de math\'ematiques d'Orsay (LMO), Universit\'e Paris-Sud, CNRS, Université Paris-Saclay}
    \centerline{15 rue Georges Cl\'emenceau, 91405 Orsay Cedex}
 \centerline{fabien.vergnet@math.u-psud.fr}
}

\medskip
{\footnotesize
 \centerline{3. Institut f\"ur Geometrie und Praktische Mathematik, RWTH Aachen University}
   \centerline{Aachen, 52062, Germany}
 \centerline{hyu@igpm.rwth-aachen.de}
}

\begin{abstract}

We derive macroscopic dynamics for self-propelled particles in a fluid. The starting point is a coupled Vicsek-Stokes system. The Vicsek model describes self-propelled agents interacting through alignment. It provides a phenomenological description of hydrodynamic interactions between agents at high density. Stokes equations describe a low Reynolds number fluid. These two dynamics are coupled by the interaction between the agents and the fluid. The fluid contributes to rotating the particles through Jeffery's equation. Particle self-propulsion induces a force dipole on the fluid. After coarse-graining we obtain a coupled Self-Organised Hydrodynamics (SOH)-Stokes system. We perform a linear stability analysis for this system which shows that both pullers and pushers have unstable modes. We conclude by providing extensions of the Vicsek-Stokes model including short-distance repulsion, finite particle inertia and finite Reynolds number fluid regime.

\end{abstract}

\medskip



\section{Introduction}

Self-organised motion is ubiquitous in nature. It corresponds to the formation of large-scale coherent structures that emerge from the many-interactions between individuals without leader. Well-known examples are bird flocks, fish schools or insect swarms. However, self-organisation also takes place at the microscopic level, for example in bacterial suspensions and sperm dynamics (see e.g. Ref. \cite{creppy2016symmetry, zhang2010collective} and the reviews \cite{elgeti2015physics, koch2011collective, ramaswamy2010mechanics}). In these cases, the environment, typically a viscous fluid, plays a key role in the dynamics. 

In this paper we investigate self-organised motion of self-propelled particles (which we will refer to as `swimmers') in a viscous fluid. The main difficulty in studying these systems comes from the complex mechanical interplay between the swimmers and the fluid. Particularly, highly non-linear interactions occur between neighbouring swimmers through the perturbations that their motions create in the surrounding fluid. 
While these interactions may be treated through far-field expansions in dilute suspensions \cite{happel2012low}, they require a much more complex treatment when the density of swimmers is high.
 Here we assume that, as a result of these swimmer-swimmer interactions, the swimmers align their direction of motion. In view of this, we adopt the Vicsek model for self-propelled particles undergoing local alignment to account for these swimmer-swimmer interactions in a phenomenological way. We then couple this model with the Stokes equation for the surrounding viscous fluid by taking into account the interactions between the swimmers and the fluid. The main goal is the derivation of macroscopic equations for this coupled system in terms of the time-evolution of the velocity of the fluid, on the one hand, and the swimmers' density and mean direction of motion, on the other hand.
 
 The coupling terms considered here coincide with the ones in the kinetic Doi-Saintillan-Shelley model, which models active and passive rod-like dilute suspensions, \cite{chen2013global,saintillan2,saintillan1}. This kinetic equation extends the Doi model \cite{doi,doibook} for liquid crystals (corresponding to passive rod-like or ellipsoidal particle suspensions) to active agents. In Ref. \cite{chen2013global}, the authors prove the existence of global weak entropic solutions for this equation and Ref. \cite{ezhilan2013instabilities} offers a closure approximation to obtain approximate macroscopic equations.  However, the Doi-Saintillan-Shelley model does not include direct swimmer-swimmer interactions as it assumes a regime with a rather low density of swimmers. The Vicsek-Stokes coupling presented here is designed to handle larger densities of active particles. An extension of the Doi-Saintillan-Shelley model for high concentration of agents can be found in Ref. \cite{ezhilan2013instabilities}; the swimmer-swimmer steric interactions are modelled through nematic interactions. Here we consider polar interactions as encompassed in the Vicsek model rather than nematic ones. Indeed, polar interactions seem more appropriate to some types of suspensions such as sperm \cite{creppy2016symmetry}. Additionally,	alignment interactions are not sufficient to prevent clustering in some high density situations. To prevent them, it is necessary  to add short-range repulsion, as we do in  Sec. \ref{sec:repulsion} following the works in Ref. \cite{degond2015macroscopic}. 

The Vicsek model \cite{vicsek2012collective} is a particle system where the position and velocity orientation of each individual particle is followed over time. It describes  self-propelled particles moving at a constant speed and trying to align their direction of motion with their neighbours, up to some noise.  There exists a variety of mathematical models for collective dynamics, see Ref. \cite[Sec. V]{marchetti2013hydrodynamics} and references therein as well as the review \cite{vicsek2012collective}. The Vicsek model is an agent based model and, consequently, a microscopic description. By contrast, Stokes equations form a continuum model for the evolution of the fluid velocity and pressure fields, which are macroscopic quantities. Therefore, the coupled Vicsek-Stokes model presented here is a hybrid microscopic/macroscopic system. This is legitimate in view of the difference in size between water molecules and the swimmers ($10^{-10}$ metres for the former, and of the order of $10^{-5}$ metres for the latter). 

The main goal of this paper is to provide a coarse-grained description of the hybrid Vicsek-Stokes dynamics in the form of a fully macroscopic description in both the fluid and the swimmers. The coarse-graining for the Vicsek model alone leads to the `Self-Organised Hydrodynamics' (SOH) equations derived in Refs. \cite{degond2015continuum,Hydro_limit}. The SOH model is a system of continuum equations for the density and mean velocity orientation of the swimmers. Here, for the first time, we provide the coarse-graining of the hybrid Vicsek-Stokes model, leading to the coupled SOH-Stokes model. The resulting model is a fully coupled model for the agents' continuum density and mean velocity orientation on the one hand and the fluid velocity and pressure fields on the other hand. The coarse-graining methodology  is based on the Generalised Collision Invariant concept introduced in Ref. \cite{degond2015continuum}. This technique has already been successfully applied to a wide range of models inspired by the Vicsek model, see Refs. \cite{degond2015phase,degond2016new,degond2017quaternions,degond2015multi}.

 The rigorous derivation of macroscopic dynamics establishes a clear link between the microscopic and macroscopic scales and, in particular, between the parameters of the two systems. Moreover, microscopic simulations tend to be very costly for a large number of individuals. Macroscopic simulations are much more cost-effective. In kinetic theory, the coarse-graining from particle dynamics to macroscopic dynamics is carried out with an intermediate step called the kinetic equation (or mean-field equation). The kinetic equation gives the distribution of a `typical particle' (if such exists) when the number of particles becomes large. Here we will derive the kinetic equation from the microscopic Vicsek-Stokes model in Sec. \ref{eq:kinetic_eq}. From the kinetic equation, we will derive then the macroscopic coupled SOH-Stokes system (Sec. \ref{sec:macro_limit}). For some general reviews on the mathematical theory of coarse-graining, the reader is referred to Refs. \cite{cercignani2013mathematical,degond2004macroscopic,sone2012kinetic}.

The complexity of the dynamics of self-organised motion in a fluid renders the rigorous macroscopic derivation extremely hard in general.  Some attempts can be found in Refs. \cite{baskaran2009statistical} and \cite[Sec. V]{marchetti2013hydrodynamics}. For the case of suspensions of passive particles, the Doi-Onsager model has been coarse-grained into the Ericksen-Leslie system, see Refs. \cite{han2014microscopic,wang2015small,weinan2006molecular}. In the case of the Cuker-Smale model (a different model for collective dynamics), it has been coupled to a Navier-Stokes equation and coarse-grained in Ref. \cite{carrillo2016analysis}.  Related works couple chemotaxis with fluid equations, see for example Ref. \cite{liu2011coupled}. A coarse-graining has been carried out for the chemotaxis-Navier-Stokes equations in Ref. \cite{zhang2016inviscid}, see also Ref. \cite{bellomo2016multiscale} for a related result. 
 
Another advantage of coarse-grained equations is that their stability analysis is far more manageable than that of microscopic models. To illustrate this effectiveness, we perform the linear stability analysis of the SOH-Stokes model. Obvious stationary solutions of the SOH-Stokes model consist of uniform (space-independent) swimmer density and mean orientation fields as well as uniform fluid velocity and pressure. We linearise the SOH-Stokes system around these stationary solutions, meaning that we consider small perturbation of a uniform state. Note that the SOH model describes aligned states as the swimmer distribution is given by a von Mises distribution. So this analysis gives access to the stability of suspensions in their aligned state only. The investigation of the stability of the isotropic state is deferred to future work. 

The stability analysis reveals that both pusher and puller types of swimmers (explained in Sec. \ref{sec:coupling_terms}) have unstable modes. This is consistent with previous studies \cite{saintillan2} which showed that both pushers and pullers are unstable to perturbations of an aligned state. Note that in \cite{saintillan2} nematic interactions were considered, while we deal with polar interactions. Additionally, the aligned state in \cite{saintillan2} is a Dirac delta in the orientation while ours is a von Mises distribution; we show that instability happens for all values of the angular dispersion around the alignment direction. We also notice that pullers are stable if they are slender rod particles. For both pushers and pullers, the instability only prevails for small $|k|$ modes (or large wavelength). The largest growth rate takes place in the limit when the mode $k\to 0$, which means that patterns induced by the instability will have roughly the same size as the system. 

Alignment interaction is not sufficient to prevent the appearance of large-concentration clusters in general \cite{czirok1996formation}. So, in cases where such clusters are not observed, it is likely that short-range repulsion effects take place in addition to alignment. Following Ref. \cite{degond2015macroscopic}, we will investigate how both the micro and macroscopic models can be extended through the introduction of a short-range repulsion force. Besides, when the particle mass and size are larger, for example for fish, it is not legitimate to neglect the particle inertia and the fluid Reynolds number any longer. Therefore, we will show how to extend the micro and macroscopic models to include such finite size effects. 

The document is structured as follows. In the next section we present the individual based model for the Vicsek-Stokes coupling and discuss the main result corresponding to its hydrodynamic limit. In Sec. \ref{sec:mean_field} we present the mean-field limit, the scaling considered, and the Generalised Collision Invariant concept. In Sec. \ref{sec:macro_limit} we prove the main result.  Sec. \ref{sec:stability} shows the stability analysis. In Sec. \ref{sec:refinements} we extend the model to account for short-range repulsion and finite inertia and Reynolds number. Finally, we conclude in Sec. \ref{sec:conclusion} discussing some perspectives on this problem.

\section{The model and discussion of the main results}
\label{sec:discussion}

\subsection{The Vicsek-Stokes coupled dynamics.} 
 
The dynamics of the viscous fluid follow Stokes equations. We couple these two models by incorporating the interaction mechanisms between the agents and the fluid. The dynamics of $N$ agents are given by the evolution of $(X_i(t), \omega_i(t))_{i\in\{1,\hdots,N\}}$ as a function of time $t\geq 0$, where $X_i(t)\in \R^3$ is the position of the $i$-th agent and $\omega_i(t)\in \mathbb{S}^{2}$ (the 2-dimensional sphere) is a unitary vector giving its direction of motion. We denote by $v=v(x,t)\in \R^3$ the fluid velocity at position $x\in\R^3$ at time $t$ and $p(x,t)\in \R$ its pressure. Here we assume that the fluid density remains constant.  In Sec. \ref{sec:derivationIBM} (see Remark \ref{rem:derivation_VS}) we derive the following Vicsek-Stokes coupled dynamics, where all the quantities are dimensionless and where the stochastic differential equation \eqref{eq:IBM2} must be understood in the Stratonovich sense, where the unknowns are $(X_i(t), \omega_i(t))_{i\in \{1,\hdots,N\}}$, $v(x,t)$, $p(x,t)$:
\begin{subequations}\label{eq:Vicsek_Stokes}
\begin{numcases}{}
			dX_i = u_i dt= v(X_i,t) dt + a\omega_i dt , \label{eq:IBM1}\\
			d\omega_i = P_{\omega_i^{\perp}}\circ \Big[\nu \overline{\omega_i}dt  + \sqrt{2D}\,dB_{t}^{i}\,\, +\big(\lambda S(v) + A(v)\big)\omega_i dt\,\,\Big], \label{eq:IBM2}\\
			\bar\omega_i = \frac{J_i}{|J_i|} \text{ with } J_i = \sum_{k=1}^N  K\lp \frac{|X_i-X_{k}|}{R} \rp\omega_k, \label{eq:IBM2.1}\\
			-\Delta_x v + \nabla_x p =  - \frac{b}{N}\sum_{i=1}^{N}\lp \omega_i\otimes\omega_i-\frac{1}{3}\Id\rp\nabla_x \delta_{X_i(t)},\label{eq:IBM3}\\ 
			\nabla_x \cdot v = 0. \label{eq:IBM4}
		\end{numcases}
\end{subequations}
In this system $a,\nu,D, \lambda,R$ and $b$ are constants. The symbol `$\otimes$' denotes the tensorial product and `$\Id$' the 3$\times$3 identity matrix. The symbol $P_{\omega_i^\perp}=\Id -\omega_i\otimes\omega_i$ gives the orthonormal projection operator onto the sphere $\mathbb{S}^2$ at $\omega_i$; the `$\circ$' symbol following it indicates that the Stochastic Differential Equation \eqref{eq:IBM2} has to be understood in the Stratonovich sense. The terms $(B^i_t)_{t\geq0}$, $i=1,\hdots, N$ are independent Brownian motions in $\R^3$. The terms $S,A$ are matrices that will be defined later. The operators $\Delta_x$, $\nabla_x$, $\nabla_x \cdot$ indicate the Laplacian, the  gradient and the divergence in $\R^3$, respectively. The symbol $\delta_X$ is the delta distribution in $\R^3$ at $X\in \R^3$. Finally, $K=K(r)\geq 0$, $r\geq 0$, is a given sensing function.
		
\medskip		
We explain first the meaning of the equations without the coupling terms. Eqs. \eqref{eq:IBM1}-\eqref{eq:IBM2.1} without the terms involving the velocity of the fluid $v$ correspond to the Vicsek model: each agent $i$ moves at a constant speed $a>0$ in the direction $\omega_i$ while trying to adopt the average direction of motion of its neighbours. This averaged direction is given by $\bar \omega_i$ in Eq. \eqref{eq:IBM2.1}.  The positive kernel $K$ weights the influence of the neighbouring agents and
 the constant $R>0$ gives the typical interaction range between agents. The intensity of alignment is given by $\nu>0$. While trying to align, agents make errors. This is modelled via a noise term $\sqrt{2D}dB^i_t$, where  $D>0$ is the standard deviation of this random motion per unit of time. The presence of projection operator $P_{\omega_{i}^\perp}$ ensures that $|\omega_i(t)|=1$ for all times (since the stochastic differential equations \eqref{eq:IBM1}-\eqref{eq:IBM2} are interpreted in the Stratonovich sense, see Ref. \cite{hsu2002stochastic}).

Eqs. \eqref{eq:IBM3} and \eqref{eq:IBM4} are the Stokes equation for the velocity of the fluid $v$ with a time-dependent force term at the right-hand  side of Eq. \eqref{eq:IBM3} whose meaning will be explained in the following section. The  hydrostatic pressure of the fluid $p=p(x,t)$ is the Lagrange multiplier of the incompressibility constraint \eqref{eq:IBM4}.

\bigskip

\noindent \paragraph{The coupling terms.} 
\label{sec:coupling_terms}

\begin{enumerate}[leftmargin=*]
\item \textbf{Effect of the fluid on the agents.}
	\begin{enumerate}[leftmargin=10pt]
		\item[i)] \textit{Effect on the agents' velocity:} In the limit of zero particle inertia there is an instantaneous relaxation of the passive part  of the particles' velocity (i.e., the particle velocity minus the self-propulsion velocity) to the velocity of the fluid (see Appendix \ref{sec:derivationIBM}). As a consequence, the term $u_i$ in Eq. \eqref{eq:IBM1}, giving the total velocity of agent~$i$, is the sum of the fluid velocity $v$ and the agent's self-propelled velocity $a\omega_i$.
		
		\item[ii)] \textit{Effect on the agents' orientation:} This is expressed by the term $\lp \lambda S(v) +A(v)\rp \omega_i $ in Eq. \eqref{eq:IBM2}, where the matrices $A$ and $S$ are  the antisymmetric and symmetric parts of the linear flow $\nabla_x v$ (which is a matrix with components $(\nabla_x v)_{ij}=\partial_{x_i}v_j$, $i,j=1,2,3$), respectively:
\begin{align}
&A(v)= \frac{1}{2}\lp\nabla_x v - (\nabla_x v)^T \rp, \label{eq:A}\\
&S(v) = \frac{1}{2}\lp\nabla_x v+ (\nabla_x v)^T \rp, \label{eq:S}
\end{align}
where the exponent `$T$' indicates the transpose of the matrix. This term encompasses Jeffery's equation, which describes the effect of a viscous fluid on a spheroidal passive particle. In a spatially homogeneous flow where $\nabla_x v$ is constant, these equations give the motion of the principal axis of spheroidal particles, as follows:
\begin{equation}\label{eq:Jeffery}
\frac{d\omega_i}{dt}= P_{\omega_i^\perp} \lp\lambda S(v)+A(v)  \rp\omega_i = \nabla_\omega\left[\lambda \frac{1}{2}\omega \cdot S\omega \right] + \frac{1}{2}(\nabla_x \times v   ) \times \omega, 
\end{equation}		
where $\nabla_\omega$ is the gradient on the sphere $\mathbb{S}^2$; $\nabla_x\times$ denotes the curl; and the symbols `$\cdot$', `$\times$' denote the inner product and the cross product in $\R^3$, respectively. The matrix $S$ describes straining forces in the fluid which forces a passive particle  to orient in a  preferred direction called `local extensional axis', given by the eigenvector of maximal eigenvalue of $S$. The matrix $A$ describes shear effects in the fluid that have the effect of rotating the suspended particle around an axis parallel to the vorticity $\nabla_x \times v$. The parameter $\lambda\in[-1,1]$ is a shape parameter: for a spheroidal particle with aspect ratio $\chi$, we have $\lambda = (\chi^2-1)/(\chi^2+1)$. The limit $\lambda \to 1$ corresponds to a slender rod-like particle, the limit $\lambda \to -1$ corresponds to a thin disk and the case $\lambda = 0$ corresponds to a sphere. For an  explanation of Jeffery's equation see e.g Refs. \cite{chen2013global,jeffery}.

	\end{enumerate}

\item \textbf{Effect of the agents on the fluid.} We consider two forces produced by agents that act on the fluid: 
		\begin{enumerate}[leftmargin=10pt]
			\item[i)] \textit{Drag force exerted on the fluid by the motion of the agents:} The motion of the agents creates a drag force on the fluid. However, in the limit of zero particle inertia this force vanishes, as detailed in Sec. \ref{sec:derivationIBM}. So we do not take it into account here. There is a symmetric effect of the fluid on the particles which in this limit produces instantaneous relaxation of the passive part of the velocity (the particle velocity minus the self-propulsion velocity) to the fluid velocity hence justifying Eq. \eqref{eq:IBM1}, see above.
			\item[ii)] \textit{The self-propulsion force:} The source term that appears on the right-hand side  of the equation for $v$ \eqref{eq:IBM3} describes the influence of the self-propulsion force of the agents on the fluid. The term $\nabla_x \delta_{X_i(t)}$ denotes the gradient of the Dirac delta $\delta_{X_i(t)}$ and 	it is defined in  weak form for any vector test function $\vec{\varphi}$ by:
			\begin{eqnarray*}
			\langle \nabla_x \delta_{X_i(t)}, \vec{\varphi} \rangle= - \langle \delta_{X_i(t)}, \nabla_x \cdot \vec{\varphi}\rangle = - \nabla_x \cdot \vec{\varphi}(X_i(t)),
			\end{eqnarray*}
where $\langle T ,\varphi\rangle$ denotes the duality bracket between a distribution $T$ and a test function~$\varphi$.
			Suppose that an agent swims by  pushing with its tail with a force $\vec{F}$ in the opposite direction of motion. Then the head also exerts a force $-\vec{F}$ on the fluid. This type of swimmer is called a `pusher'. If the centre of the swimmer is in location $X_i$ and it has a length $\ell$, then the pushing force is applied at location $X_i-\frac{\ell}{2}\omega_i$ while the force of the head is at $X_i+\frac{\ell}{2}\omega_i$, see Fig. \ref{Fig:pusher}. Since $\vec{F}$ and $-\vec{F}$ are applied at different points they do not cancel, this is referred as a `force dipole' and in this case
			$$\vec{F}= |\vec{F}|\omega_i \left[\delta_{X_i+\frac{\ell}{2}\omega_i}-\delta_{X_i-\frac{\ell}{2}\omega_i}\right].$$
			From a Taylor expansion for small $\ell$, this leads to the right-hand side term in Eq. \eqref{eq:IBM3}, with $b>0$ (after dividing by the fluid viscosity). We do not give here the details of this derivation but refer the reader to Refs. \cite{chen2013global,saintillan2,saintillan1} and references therein. Another common way of swimming is by using the arms. Then the swimmer is called `puller' and in this case $b<0$ in Eq.\eqref{eq:IBM3}. For a drawing of a puller see Ref. \cite[Fig. 3]{chen2013global}.
			
\begin{figure}
\centering
\includegraphics[scale=0.5]{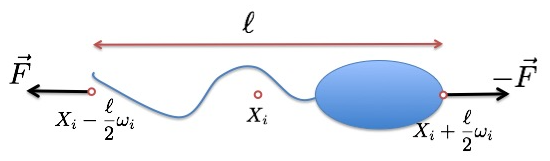}
\caption{Shape of a pusher. The tail exerts a force $\vec{F}$ on the fluid in the direction of the motion $\omega_i$ while the head exerts a force $-\vec{F}$. Since these two forces are not applied on the some point, they do not cancel.}
\label{Fig:pusher}
\end{figure}

		\end{enumerate}

\item \textbf{Other interactions and effects.} The model presented here is a simplification of the actual dynamics. Therefore, it could be enriched by taking into consideration other mechanical effects, like noise in the spatial variable (see the Doi-Saintillan-Shelley model \cite{chen2013global}) and extra forces acting on the fluid due to the inextensibility of the particles (resistance to stretching and compression) when the particle sizes are not supposed infinitesimally small, see Ref. \cite{ezhilan2013instabilities}. We will consider refined versions of this model in Sec. \ref{sec:repulsion} where we add volume exclusion between agents, or in Sec. \ref{sec:derivationIBM} where  we include both fluid and particle inertia.
\end{enumerate}

As a by-product, Sec. \ref{sec:derivationIBM} provides a derivation of the Vicsek-Stokes system \eqref{eq:IBM1}-\eqref{eq:IBM4} from a Vicsek-Navier-Stokes coupling.

\subsection{Macroscopic coupled dynamics: the SOH-Stokes model.}
In Sec. \ref{sec:macro_limit} (Th. \ref{th:hydro_limit}), from the Vicsek-Stokes dynamics \eqref{eq:IBM1}-\eqref{eq:IBM4}, we derive the following macroscopic system , that we refer to as the `Self-Organised Hydrodynamics-Stokes model' (SOH-Stokes). It gives  the time-evolution of the spatial density of agents $\rho=\rho(x,t)$,  the mean direction of motion $\Omega=\Omega(x,t)$, the velocity of the fluid $v=v(x,t)$ and the pressure~$p=p(x,t)$:
\begin{subequations}\label{eq:macro_SOH_Stokes}
\begin{numcases}{}
		 \partial_t\rho + \nabla_x\cdot(\rho U)  = 0, \label{eq:continuity_equation} \\
		 \rho\partial_{t}\Omega + \rho (V\cdot\nabla_x)\Omega +\frac{a}{\kappa} P_{\Omega^{\perp}}\nabla_x\rho= \gamma P_{\Omega^{\perp}}\Delta_x(\rho\Omega) +\rho P_{\Omega^{\perp}} \lp\tilde\lambda S(v) +A(v)\rp \Omega ,\qquad \label{eq:hydro2} \\
		-\Delta_x v + \nabla_x p =  - b\nabla_x\cdot\lp \rho \mathcal{Q}(\Omega) \rp, \label{eq:hydro3}\\ 
		\nabla_x\cdot v=0 \label{eq:hydro4}, 
	\end{numcases}
\end{subequations}
	where 
	$$
		U = a c_1 \Omega +v, \quad
		V = a c_2 \Omega + v, \quad
		\mathcal{Q}(\Omega) = c_4\lp \Omega \otimes \Omega-\frac{1}{3}\Id \rp, 
		$$
		\begin{equation}
				\gamma = k_0\nu \lp c_2 + \frac{2}{\kappa} \rp, \quad
		\tilde \lambda = \lambda \lambda_0, \quad \lambda_0= \frac{6}{\kappa}c_2 + c_3-1 , \label{eq:lambda0}
	\end{equation}
	and where the constants $c_1,\hdots, c_4$ and $k_0$ are given by Eqs. \eqref{eq:c1_th}--\eqref{eq:c4}, \eqref{eq:def_k0}, and where $\kappa=\nu/D$.

	The first two Eqs. \eqref{eq:continuity_equation}-\eqref{eq:hydro2} provide the time-evolution of the density of the agents $\rho=\rho(x,t)$ and their mean direction of motion $\Omega=\Omega(x,t)$, respectively. Without the terms involving the velocity of the fluid $v$, Eqs. \eqref{eq:continuity_equation}-\eqref{eq:hydro2} correspond to the Self-Organised Hydrodynamics (SOH) macroscopic equations for the Vicsek model, obtained in Refs. \cite{degond2015phase,degond2008continuum,Hydro_limit} (the diffusive term is derived in Ref. \cite{degond2015phase}). The SOH system resembles a fluid dynamics equation, particularly, a compressible Navier-Stokes equation, with the differences that $c_1\neq c_2$ and the `velocity' $\Omega$ is constrained to be of norm one with the presence of the projection operator $P_{\Omega^\perp}$ in the pressure $\nabla_x\rho$ and the diffusion $\Delta_x(\rho\Omega)$. This projection operator precisely ensures that $|\Omega|=1$ at all times (provided that $|\Omega|_{t=0}=~1$), but, as a consequence, the equation is not conservative, meaning that the terms involving spatial derivatives cannot be written as the spatial divergence of a flux function.
Eqs. \eqref{eq:hydro3}--\eqref{eq:hydro4} without the right-hand side in \eqref{eq:hydro3} correspond to the Stokes equation.

\paragraph{The coupling terms.} 		
		The terms involving the velocity of the fluid $v$ in Eqs.  \eqref{eq:continuity_equation}-\eqref{eq:hydro2} express the effect of the fluid on the particles. As expected, the continuity equation \eqref{eq:continuity_equation} for the density $\rho$ has a velocity $U$ which is the sum of the local average self-propulsion velocity ($cc_1\Omega$) and the velocity of the fluid $v$. Also the convective velocity in Eq. \eqref{eq:hydro2} for~$\Omega$ is a weighted sum of $\Omega$ and the velocity of the fluid. The last term in Eq. \eqref{eq:hydro2} reflects how Jeffery's equation on individual swimmers \eqref{eq:Jeffery} translate to the population level. The terms resulting from Jeffery's equation describe the propensity of $\Omega$  to align in the direction of the so-called local extensional axis, given by the eigenvector of maximal eigenvalue of $S$, as well as to rotate about the vorticity axis, parallel to $\nabla_x \times v$. Notice, though, that the shape parameter $\tilde\lambda \neq \lambda$. Numerically, (see Fig. \ref{Fig:lambda}) we observe that $\lambda_0\in[0,1]$. This implies that $\tilde\lambda\in[-1,1]$; $\tilde \lambda$ has the same sign as $\lambda$; and $|\tilde \lambda|\leq |\lambda|$. Therefore, Jeffery's equation for the agents' individual orientations is coarse-grained into another Jeffery's equation  for the local mean  orientation, but the `mean particle shape' associated with $\tilde \lambda$ is different from the individuals' shapes associated with $\lambda$. Particularly, when $\kappa\to 0$ (large noise regime), we get $\tilde\lambda=0$, which corresponds to the shape of a sphere, and when $\kappa\to \infty$ (low noise regime), we get $\tilde\lambda=\lambda$, and we recover the original shape parameter.
		
Finally, the right-hand side in Eq. \eqref{eq:hydro3} gives the influence of the agents on the fluid. It involves the divergence of the deviatoric stress tensor $\mathcal{Q}(\Omega)$, i.e., the contribution of the swimmers to the extra-stress, which provides its non-Newtonian character to the fluid. This term results from  the coarse-graining of the right hand side of Eq. \eqref{eq:IBM3}. Numerically we observe (Fig. \ref{Fig:c4}) that $c_4>0$, which shows that the coarse-graining of a population of pushers preserves the `pusher' behaviour, as it should.

The reader is referred to Sec. \ref{sec:refinements} for some extensions of this model.

\begin{figure*}[t!]
\begin{subfigure}[b]{0.5\textwidth}
\includegraphics[width=\textwidth]{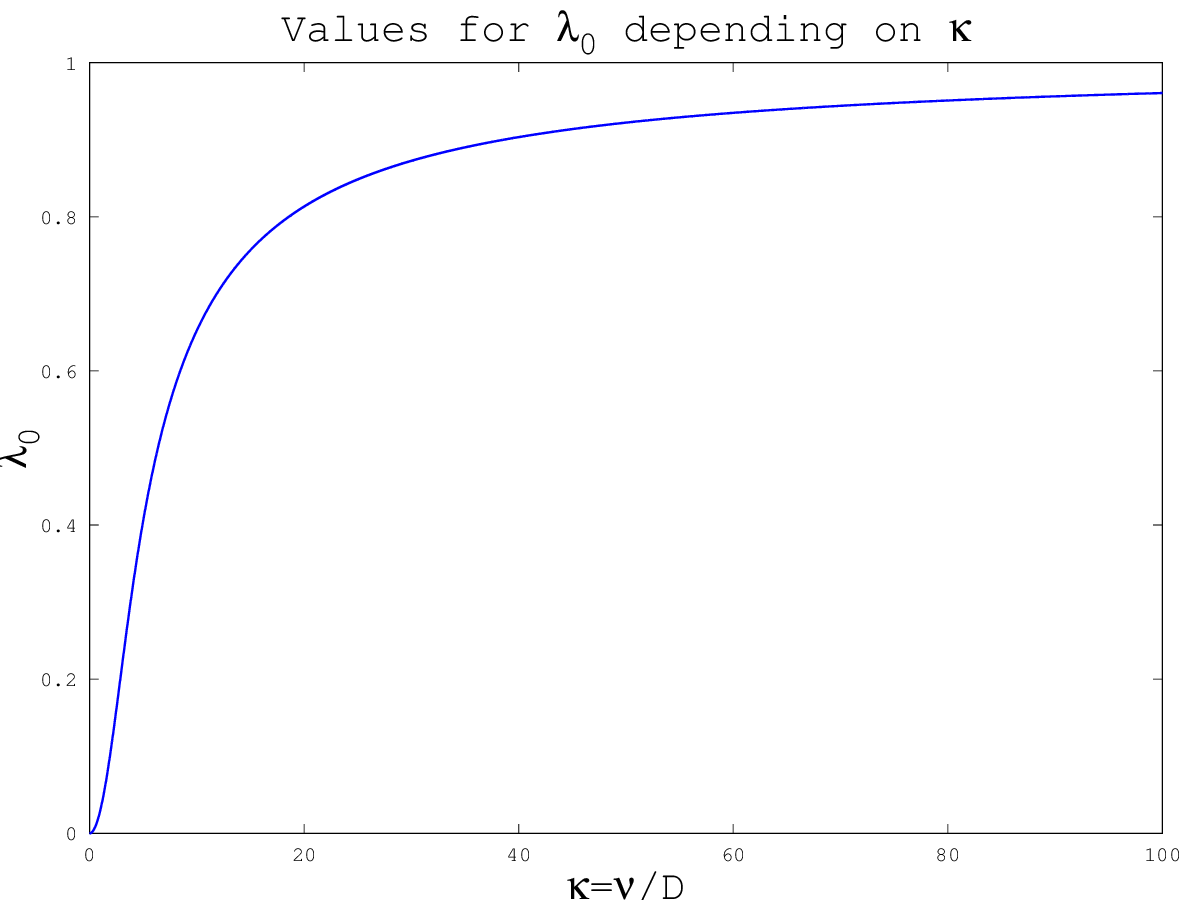}
\caption{ }
\label{Fig:lambda}
\end{subfigure}
\hfill
\begin{subfigure}[b]{0.5\textwidth}
\includegraphics[width=\textwidth]{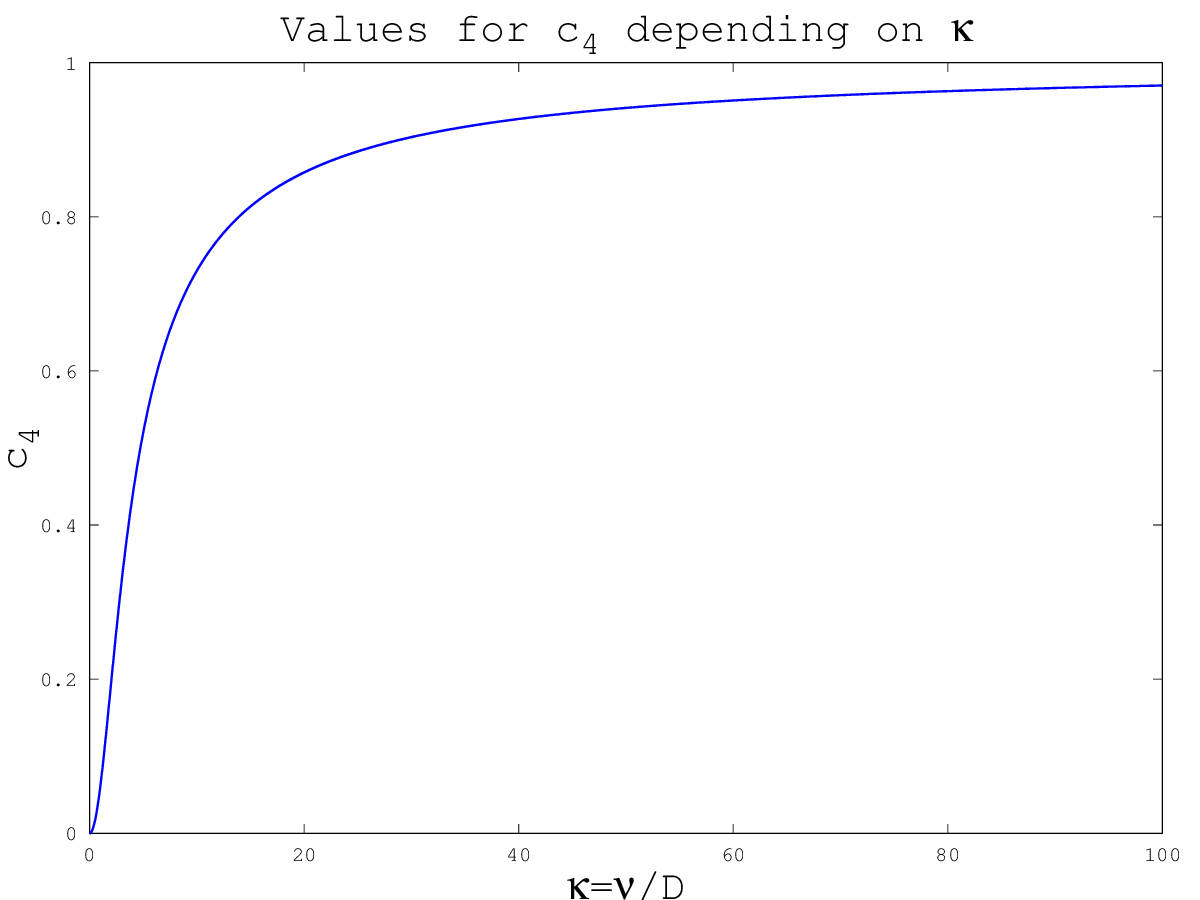}
\caption{ }
\label{Fig:c4}
\end{subfigure}
\caption{\ref{Fig:lambda}: Plot of the values of $\lambda_0$ in Eq. \eqref{eq:lambda0} as a function of $\kappa$ (done in dimension 2, where the generalised collision invariant has an explicit form \cite{frouvelle2012continuum}, corresponding to $\lambda_0=\frac{4}{\kappa}c_2+c_3-1$). \ref{Fig:c4}. Plot of the values of $c_4$ in Eq. \eqref{eq:c4}.}
\end{figure*}

\section{The mean-field limit equation}
 \label{sec:mean_field}
  
 As explained in the introduction, the derivation of the macroscopic equations is carried out with an intermediate step: the kinetic or mean-field equations.
The mean-field limit of System \eqref{eq:IBM1}-\eqref{eq:IBM4} provides the time-evolution of the distribution function  $f=f(x,\omega,t)$ in space and orientation  of a typical agent.  From the equation on $f$, we will derive the macroscopic equations in Sec. \ref{sec:macro_limit}.
For the case of the Vicsek model alone, a rigorous proof of the mean-field limit has been obtained in Ref. \cite{bolley2012mean} when there is no normalisation of $\bar \omega_i$ in Eq. \eqref{eq:IBM2.1}, i.e., when $\bar \omega_i=J_i$.	Following the proof in \cite{bolley2012mean} formally we have the:
\begin{proposition}[(Formal) Mean-field limit]
\label{prop:mean-field}
Consider the empirical distribution associated to the dynamics of the agents in Eqs. \eqref{eq:IBM1}-\eqref{eq:IBM4}, i.e.:
\begin{equation} \label{eq:empirical_distribution}
			f^{N}(x,\omega,t)=\frac{1}{N}\sum_{i=1}^{N}\delta_{x_i(t)}(x)\otimes\delta_{\omega_i(t)}(\omega),
		\end{equation}
		where $\delta_{x_i(t)}(x)$ and $\delta_{\omega_i(t)}(\omega)$ denote the Dirac delta at $x_i(t)$ and $\omega_i(t)$ on $\R^3$  and $\mathbb{S}^2$, respectively. Assume that  $f^N$ converges weakly to $f=f(x,\omega,t)$ as the number of agents $N\to \infty$. Then, the limit $f$ satisfies the following system:
		\begin{equation}
			\begin{cases}
			&\partial_t f + \nabla_x\cdot(u_{(f,v)}f) + \nabla_{\omega}\cdot\Big( \big[P_{\omega^{\perp}} \left\{\nu \overline{\omega}_{f} +\lp\lambda S(v)+A(v)\rp\omega\right\}\big]f\Big) = D\Delta_{\omega}f, \label{eq:kinetic_eq}\\
			&u_{(f,v)}(x,\omega,t)=v(x,t) + a\omega, \\
			&-\Delta_x v + \nabla_x p =  - b\nabla_x\cdot Q_{f},\\
			&\nabla_x \cdot v =0,
			\end{cases}
		\end{equation}
		where  $\nabla_{\omega}\cdot$ and $\Delta_{\omega}$ stand for the divergence  and the laplacian in $\mathbb{S}^2$, respectively; and where
		\begin{align}
			&\rho_{f}(x,t)=\int_{\mathbb{S}}{f(x,\omega,t)d\omega},\label{eq:density_local}\\
			&\overline{\omega}_{f}(x,t)=\frac{J_{f}(x,t)}{|J_{f}(x,t)|},\label{eq:bar-omega-f}\\
			&J_{f}(x,t)= 	\int_{\mathbb{S}^2\times\mathbb{R}^{3}}{K\left(\frac{|x-y|}{R}\right)\omega f(y,\omega,t) \,dyd\omega}, \label{eq:J_kinetic}\\
			&Q_{f}(x,t)=\int_{\mathbb{S}^{2}}{\lp \omega\otimes\omega-\frac{1}{3}\mbox{\emph{Id}}\rp f(x,\omega,t)\,d\omega}.
		\end{align}

\end{proposition}

\subsection{Scaling and expansion}
 
We scale the alignment intensity and the variance of the noise  by setting $\nu=\tilde\nu/\eps$, $D=\tilde D/\eps$, where $\tilde \nu$, $\tilde D$ are given fixed quantities. Considering the classical Vicsek model (without the coupling terms), this rescaling corresponds to
$$d\omega_i=  P_{\omega_i^\perp}\lp \frac{\nu}{\eps} \bar \omega_i dt+ \sqrt{2\frac{D}{\eps}}dB^i_t \rp=P_{\omega_i^\perp}\lp \nu \bar \omega_i d(t/\eps)+ \sqrt{2D}dB^i_{t/\eps} \rp,$$
i.e., it corresponds to a time-rescaling $t'=t/\eps$ which, as $\eps\to 0$, gives the long-time dynamics for $\omega_i$. In order words, with this rescaling we express the fact that the self-propulsion velocity of the agents $\omega$ is a fast-varying variable while the velocity of the fluid $v$ is a slow-varying variable. Notice, however, the invariance of the quotient
\begin{equation} \label{eq:kappa}
\kappa:=\frac{\nu}{D} = \frac{\tilde \nu}{\tilde D},
\end{equation}
that we denote by $\kappa$.
We also scale the radius of influence $R$ in Eq. \eqref{eq:J_kinetic} by setting $R=\sqrt{\eps} \tilde R$, which is the rescaling considered in Ref. \cite{degond2015phase}. This rescaling expresses that the interactions between agents become localized in space as $\eps \to 0$. 
After rescaling the kinetic equation \eqref{eq:kinetic_eq} in this way, we obtain (after skipping the tildes):
		\begin{equation}
			\begin{cases} \label{eq:syst1}
			&\varepsilon\left[\partial_t f^{\varepsilon} + \nabla_x\cdot(u_{(f^{\varepsilon},v^\eps)}f^{\varepsilon})\right] + \nabla_{\omega}\cdot\Big( \big[P_{\omega^{\perp}} \{ \nu\overline{\omega}_{f^{\varepsilon}}+\varepsilon\lp\lambda S(v^\eps)+A(v^\eps)\rp\omega\}\big]f^{\varepsilon}\Big) = D\Delta_{\omega}f^{\varepsilon}, \\
			&u_{(f^{\varepsilon},v^\eps)}(x,\omega,t)=v^\eps(x,t) + a\omega, \\
			& \bar \omega^\eps_{f}= \frac{J^\eps_{f}}{|J^\eps_{f}|}, \quad J^\eps_{f} = \int_{\mathbb{S}^2 \times \R^3} \omega K\lp\frac{|x-y|}{\sqrt{\eps}R} \rp f\, d\omega dy, \\
			&-\Delta_x v^\eps + \nabla_x p^\eps =  - b\nabla_x\cdot G_{f^{\varepsilon}},\\
			&\nabla_x\cdot v^\eps=0
			\end{cases}
		\end{equation}
We simplify this system by considering the following expansion:
\begin{lemma}
\label{lem:expansion}
It holds that
\begin{equation} \label{eq:expansion-bar-omega}
\bar \omega^\eps_{f}= \Omega_{f} + \eps\frac{k_0}{|j_{f}|} {P_{\Omega_{f}^\perp}}\Delta_x j_{f}+ \mathcal{O}(\eps^2),
\end{equation}
where
	\begin{equation} \label{eq:def_k0}
			k_{0}=\frac{R^{2}}{6}\int_{\mathbb{R}^{3}}{K(|x|)|x|^{2}\,dx} \lp\int_{\R^3}K(|x|) \, dx \rp^{-1}, 
		\end{equation}
and
	\begin{align}
		j_{f}(x,t)&=\int_{\mathbb{S}^2}\omega{f(x,\omega,t) \,d\omega} \qquad \mbox{(local current density)},\label{eq:localcurrent}\\
			\Omega_{f}(x,t)&=\frac{j_{f}(x,t)}{|j_{f}(x,t)|}\label{localaverageorient} \qquad \mbox{(local average orientation)}.
		\end{align}
\end{lemma}	

\begin{proof}
The result is a direct consequence of the Taylor expansion for 
$$J^\eps_{f}(x,t)=  \int_{\mathbb{S}^2\times \mathbb{R}^3} \omega \,K\lp\frac{|x-y|}{\sqrt{\eps} R} \rp f(y,\omega, t)   \, d\omega dy,$$
after performing the change of variables $z=(x-y)/(\sqrt{\eps}R)$,
which gives:
\begin{eqnarray*}
 J^\eps_{f}&=&(\sqrt{\eps} R)^3 \int_{\R^3}K(|x|)\,dx \, \big( j_{f}+\eps k_0 \Delta_x j_{f} + \mathcal{O}(\eps) \big),\\
 |J^\eps_{f}|^{-1}&=&\left[(\sqrt{\eps} R)^3 \int_{\R^3}K(|x|)\,dx\right]^{-1} \, |j_{f}|^{-1}\lp 1- \eps k_0 (j_{f}\cdot \Delta_x j_{f})|j_{f}|^{-2} \rp+\mathcal{O}(\eps^2).
 \end{eqnarray*}
\end{proof}	
		
Thanks to the previous Lemma \ref{lem:expansion}, we can rewrite the rescaled system \eqref{eq:syst1} as follows: 		
\begin{empheq}[left=\empheqlbrace]{align}\label{eq:scaledsys}
			&\varepsilon\big[\partial_tf^{\varepsilon} + \nabla_x\cdot(u_{(f^{\varepsilon},v^\eps)}f^{\varepsilon}) + \nabla_{\omega}\cdot(\mathcal{F}_{(f^{\varepsilon},v^\eps)}f^{\varepsilon})\big] = Q(f^{\varepsilon}) + \mathcal{O}(\eps^2),\\
			&u_{(f^{\varepsilon},v^\eps)}=v^\eps(x,t) + a\omega, \\
			&-\Delta_x v^\eps + \nabla_x p^\eps =  - b\nabla_x\cdot G_{f^{\varepsilon}}, \label{eq:scaledsys_v}\\
			&\nabla_x\cdot v^\eps=0, \label{eq:scaledsys_end}
		\end{empheq}
with
				\begin{align}
			Q(f)&=-\nabla_{\omega}\cdot\left[\nu P_{\omega^\perp}(\Omega_{f}) f\right] + D\Delta_{\omega}f, \label{eq:Q}\\
			\mathcal{F}_{(f,v)}&=P_{\omega^{\perp}}\left[\nu\frac{k_0}{|j_{f}|} P_{\Omega_{f}^\perp}\Delta_x  j_{f}  + \lp \lambda S(v)+A(v) \rp\omega\right],
		\end{align}
		where $j_{f}, \Omega_{f}$ are given in Eqs. \eqref{eq:localcurrent}-\eqref{localaverageorient} .

\subsection{Equilibria and Generalised Collision Invariants}
								
In Ref. \cite{degond2008continuum} the authors studied the operator $Q$ given in Eq. \eqref{eq:Q}. They proved that it can be recast into  a Fokker-Planck form:
$$Q(f)=D \nabla_\omega \cdot \left[ M_{\Omega_f}(\omega) \nabla_\omega \lp\frac{f}{M_{\Omega_f}(\omega)} \rp\right],$$
where the density on the sphere
$$M_{\Omega}(\omega)= \frac{1}{Z} \exp\left(\kappa(\Omega \cdot \omega)\right), \quad \int_{\mathbb{S}^2} M_{\Omega}(\omega)\, d\omega =1,$$
is the so-called von Mises distribution ($Z$ is a normalizing constant).
The equilibria of $Q$ as a function of $\omega$ are given by the set of functions
\begin{equation} \label{eq:kernelQ}
\mbox{Ker }{Q}=\{\rho M_\Omega(\omega),\, \rho\geq 0, \, \Omega \in \mathbb{S}^2 \}.
\end{equation}

Moreover, in Ref. \cite{degond2008continuum} it is proven that
\begin{equation} \label{eq:consistency relationship}
\int_{\mathbb{S}^2}\omega M_{\Omega}(\omega) \, d\omega = c_1 \Omega,
\end{equation}
for
	\begin{equation} \label{eq:def_c1}
		c_{1}:=\int_{\mathbb{S}^2} (\omega\cdot \Omega) M_{\Omega}(\omega)\, d\omega =\frac{\int^\pi_0 \cos\theta\, \exp(\kappa\cos\theta)\, \sin \theta\, d\theta}{\int^\pi_0  \exp(\kappa\cos\theta)\, \sin \theta\, d\theta} \in [0,1],
	\end{equation}
	showing the consistency relationship 
	$$\Omega_{M_\Omega} = \frac{c_1 \Omega}{|c_1\Omega|}=\Omega.$$ 
Details can be found in Eq. \eqref{eq:change of variables}.

\medskip
Collision invariants are fundamental in the derivation of macroscopic equations. They are defined as the scalar  functions $\psi$ such that
\begin{equation} 
\label{eq:collision_invariant}
\int_{\mathbb{S}^2} Q(f)(\omega)\, \psi(\omega)\, d\omega=0.
\end{equation}
In the present case, $\psi=$constant clearly satisfies this relation. This is a consequence of the conservation of mass during the interactions between agents. It can be shown that there are no other conserved quantities. This implies, particularly, that the dimension of the space of collision invariants is smaller than the dimension of the kernel $Q$ in \eqref{eq:kernelQ}, which  is 3-dimensional. Classical methods require the dimension of the two spaces to be the same in order to derive a full system of macroscopic equations. The collision invariant corresponding to the constants will allow us to derive the equation for the spatial density $\rho=\int f d\omega$ (as we will see in the next section), but it will not be enough to determine the equation for the mean orientation $\Omega$. To sort out this problem, the authors in Ref. \cite{degond2008continuum} introduce the concept of Generalised Collision Invariant (GCI) defined as follows:
\begin{definition}
A function $\psi : \mathbb{S}\to \R$ is called `Generalised Collision Invariant' associated to $\Omega_0\in\mathbb{S}^2$ if and only if
\begin{equation}
\int_{\mathbb{S}^2}\mathcal{Q}(f, \Omega_0) \psi\, d\omega = 0,\quad \mbox{for all } f \mbox{ such that } P_{\Omega_0^\perp}\lp\int_{\mathbb{S}^2}\omega f\, d\omega \rp=0,
\label{eq:GCI_def}
\end{equation}
where
$$\mathcal{Q}(f,\Omega_0) = \nabla_\omega \cdot \left[ M_{\Omega_0}(\omega) \nabla_\omega\lp\frac{f}{M_{\Omega_0}(\omega)} \rp\right].$$
\end{definition}
Notice that with this definition
$$\mathcal{Q}(f,\Omega_f)=Q(f).$$

It has been proven that the GCI has the following properties:
\begin{proposition}[Generalised Collision Invariant, from Ref. \cite{degond2008continuum}] 
\label{lem:GCI}
(i) Given $\Omega_0\in \mathbb{S}^2$, the \emph{vector GCI} defined by:
$$\vec{\psi}_{\Omega_0}(\omega)=(\Omega_0\times\omega) h(\omega\cdot \Omega_0),$$ 
satisfies (\ref{eq:GCI_def} (componentwise),  where the function $h: \R \to \R$ satisfies $h(\mu)=(1-\mu^2)^{-1/2}g \geq 0$ for $g$  the unique solution  in the weighted $H^1$ space $V$ given by
$$V= \left\{g\, |\, (1-\mu^2)^{-1/2}g\in L^2(-1,1), \quad (1-\mu^2)^{1/2}\partial_\mu g\in L^2(-1,1) \right\},$$
of the differential  equation
$$-(1-\mu^2) \partial_\mu\lp e^{\kappa\mu}(1-\mu^2)\partial_\mu g\rp+e^{\kappa\mu}g=-(1-\mu^2)^{3/2}e^{\kappa\mu}.$$
(ii) The set of GCIs associated to $\Omega_0$ consists of all functions $\psi$ such that there exist $B\in\R^3$, $B\cdot \Omega_0=0$ and $C\in \R$ such that  $\psi(\omega)=B \cdot \vec\psi_{\Omega_0}+C$.\\
(iii) For a given function $f: \mathbb{S}^2 \to \R$, we consider the associated $\Omega_f$ given by 
$$\Omega_f=\frac{j_f}{|j_f|},$$
and consider 
\begin{equation} \label{eq:GCI}
\vec\psi_{\Omega_f}(\omega)= (\Omega_f\times \omega) h(\Omega_f\cdot \omega).
\end{equation}
Then $\vec \psi_{\Omega_f}$ satisfies
\begin{equation}
\label{eq:keyCGI}
\int_{\mathbb{S}^2} Q(f)\,  \vec \psi_{\Omega_f} \, d\omega =0.
\end{equation}
\end{proposition}

\section{Macroscopic limit: the SOH-Stokes system}
\label{sec:macro_limit}

In this section we investigate the hydrodynamic limit as $\varepsilon \to 0$ for the system \eqref{eq:scaledsys}--\eqref{eq:scaledsys_end}. 
We will use the following change of variables: for $\Omega \in \mathbb{S}^2$ fixed, we decompose any given vector $\omega\in \mathbb{S}^2$ uniquely as
	\begin{equation}	\label{eq:change of variables}
	\omega = P_{\Omega}(\omega) + P_{\Omega^\perp}(\omega)= \cos\theta\, \Omega + \sin\theta \, w, \mbox{ for } w\in \mathcal{S}:= (\mathbb{S}^2 \cap\Omega^\perp) \cong \mathbb{S}^1\mbox{ and } \theta \in[0,\pi].
	\end{equation}
We take the convention $\int_{\mathbb{S}^2}d\omega=\int_{\mathcal{S}}{dw}=1$. One can check that (see Ref. \cite[Ap. A2]{frouvelle2012continuum}) for any measurable function $a(\omega)=\bar a(\theta,w)$:
	\begin{equation} \label{eq:decomposition_integral}
		\int_{\mathbb{S}^2}{a(\omega)\, d\omega}=\frac{1}{2}\int_{0}^{\pi}{\int_{\mathcal{S}}{\bar a(\theta,w)\sin\theta \,dwd\theta}},
	\end{equation}
and
	\begin{equation} \label{eq:tensorial_w}
		\int_{\mathcal{S}}{w\, dw}=0 \textrm{,   ~~\mbox{and}~~   } \int_{\mathcal{S}}{w\otimes w\,dw}=\frac{1}{2}(\mbox{Id}-\Omega\otimes\Omega)= \frac{1}{2}P_{\Omega^\perp}.
	\end{equation}
 We will also use the notations:
\begin{equation} \label{eq:notation_theta}
\tilde h(\theta):= h(\cos\theta)=h(\omega \cdot \Omega), \quad \tilde M(\theta):=M_{\Omega}(\cos \theta)=M_\Omega(\omega\cdot \Omega),
\end{equation} 
where $h$ is the function appearing in Eq. \eqref{eq:GCI}.
\begin{theorem}[(Formal) macroscopic limit] 
\label{th:hydro_limit}
Consider the rescaled system \eqref{eq:scaledsys}--\eqref{eq:scaledsys_end}. When $\eps \to 0$, it holds (formally) that
$$(f^\eps, v^\eps, p^\eps)\to ( f = \rho M_\Omega, v, p),$$
where $\rho=\rho(x,t)\geq 0$ and $\Omega = \Omega (x,t)\in \mathbb{S}^{2}$ are the limits of the local density $\rho^\eps=\int_{\mathbb{S}^2}f^\eps\, d\omega$ and 
the local mean orientation $\Omega_{f^\eps}$ in Eq. \eqref{localaverageorient}, respectively.
Moreover, if the convergence is strong enough and $\Omega$, $\rho$, $v$ and $p$ are smooth enough, they satisfy the coupled system \eqref{eq:continuity_equation}-\eqref{eq:hydro4} with explicit constants
	\begin{eqnarray}
		c_1 &=& \langle \cos\theta\,\rangle_{\sin\theta \tilde M(\theta)}, \label{eq:c1_th}\\
		c_2 &=& \langle  \cos\theta \rangle_{\sin^3 \theta \tilde M(\theta) \tilde h(\theta) }, \label{eq:c2}\\
		c_3 &=& 2\langle \cos^{2}\theta \rangle_{\sin^3 \theta \tilde M(\theta) \tilde h(\theta)},\label{eq:c3}\\
		c_4&=&1-\frac{3}{2}\left\langle  \sin^2(\theta) \right\rangle_{\sin\theta \tilde M(\theta)}  , \label{eq:c4}
	\end{eqnarray}
	where we used the following notation: for any functions $g,\ell:[0,\pi]\to \R$ define
	$$\langle g \rangle_\ell :=\int^\pi_0 g(\theta) \frac{ \ell(\theta)}{\int^\pi_0 \ell(\theta')\, d\theta'} \, d\theta.$$
	The constants $a$, $b$, $\kappa=D/\nu$ correspond to the ones in the individual based model \eqref{eq:IBM1}-\eqref{eq:IBM4} and the value of $k_0$ is given in Eq. \eqref{eq:def_k0}. 
\end{theorem}

\begin{proof}

	Suppose that $f^{\varepsilon}$  converges to $f$ as $\varepsilon\to 0$. Then, from Eq. \eqref{eq:scaledsys}, $f$ belongs to the kernel of $Q$, i.e., $Q(f)=0$. Therefore, $f=\rho M_{\Omega}$ by Eq. \eqref{eq:kernelQ}, with $\rho= \rho(x,t)\ge 0$ and $\Omega=\Omega(x,t)\in\mathbb{S}^2$. We start by computing the equations for these two macroscopic quantities.

\medskip	
	
  We obtain the continuity equation \eqref{eq:continuity_equation} for $\rho$  by integrating the kinetic equation \eqref{eq:scaledsys} with respect to $\omega$; dividing by $\eps$; taking the limit $\varepsilon \to 0$; and using the consistency relationship in Eq. \eqref{eq:consistency relationship}. Notice that the integral of the right hand side of the kinetic equation \eqref{eq:scaledsys} vanishes since $\psi=1$ is a collision invariant in Eq. \eqref{eq:collision_invariant} , i.e.,
  $$\int_{\mathbb{S}^2} Q(f^\eps)\, d\omega=0.$$

\medskip

	We compute next Eq. \eqref{eq:hydro2} for the mean direction of the agents $\Omega$. We multiply the kinetic equation \eqref{eq:scaledsys} by $\frac{1}{\eps}\psi_{f^\eps}$, where $\psi_{f^\eps}=h(\omega\cdot\Omega_{f^{\varepsilon}})(\Omega_{f^{\varepsilon}}\times\omega)$ is the Generalised Collision Invariant given by Prop. \ref{lem:GCI}, and integrate with respect to $\omega$:
$$
	\int_{\mathbb{S}^2} \big[\partial_t f^\eps + \nabla_x \cdot (u_{(f^\eps,v^\eps)} f^\eps)  + \nabla_\omega \cdot \lp \mathcal{F}_{(f^\eps,v^\eps)} f^\eps\rp\big]\, h(\omega\cdot\Omega_{f^{\varepsilon}})\,(\Omega_{f^{\varepsilon}}\times\omega) \, d\omega
	  =\mathcal{O}(\eps).
$$
	 Notice that the term involving $Q$ vanishes thanks to Eq. \eqref{eq:keyCGI}. Taking the limit $\varepsilon\to 0$ on the previous expression, we obtain:
$$
		\Omega\times X=0,\quad X:=\int_{\mathbb{S}^2}{\big[\partial_t (\rho M_{\Omega}) + \nabla_x\cdot(u_{(\rho M_{\Omega}, v)}\rho M_{\Omega}) + \nabla_{\omega}\cdot(\mathcal{F}_{(\rho M_\Omega,v)}\rho M_{\Omega})\big]\, h(\omega\cdot \Omega)\, \omega\, d\omega},
$$
	or, equivalently,
		\begin{equation} \label{eq:decomposition_X}
	P_{\Omega^{\perp}}X= 0.
	\end{equation}

To compute this last expression we decompose $X$ into $X=X_1+X_2+X_3+X_4$ for
\begin{align*}
X_1 &=\int_{\mathbb{S}^2} \left[ \partial_t (\rho M_\Omega) + a\omega \cdot \nabla_x (\rho M_\Omega) \right]\, h(\omega\cdot \Omega) \omega\, d\omega,\\
X_2 &= \int_{\mathbb{S}^2} \nabla_\omega \cdot \lp \nu P_{\omega^\perp} \frac{k_0}{|j_{\rho M_\Omega}|} P_{\Omega^\perp}( \Delta_x j_{\rho M_\Omega}) \, \rho M_{\Omega}\rp\, h(\omega\cdot \Omega)\, \omega\, d\omega,\\
X_3 &= \int_{\mathbb{S}^2}  v\cdot\nabla_x ( \rho M_\Omega)\, h(\omega \cdot \Omega) \, \omega\, d\omega,  \\
X_4 &= \int_{\mathbb{S}^2} \nabla_\omega \cdot \lp P_{\omega^\perp}\left[(\lambda S(v)+A(v))\omega \right]\, \rho M_{\Omega}\rp \, h(\omega \cdot \Omega) \, \omega\, d\omega.
\end{align*}
Notice that in $X_3$ we used the incompressibility condition $\nabla_x\cdot v=0$ to express $\nabla_x~\cdot~ ~(~v\rho M_\Omega)=v\cdot \nabla_x (\rho M_\Omega)$.

The term $P_{\Omega^\perp}X_1$ has been computed in Ref. \cite{degond2008continuum}:
$$P_{\Omega^\perp}X_1= \rho(C_0\kappa\,\partial_t \Omega + a\kappa C_2 (\Omega \cdot \nabla_x) \Omega) + aC_0 P_{\Omega^\perp} \nabla_x\rho ,$$
for
$$C_0 =\frac{1}{4}\int^\pi_0 \sin^3\theta\, \tilde h(\theta)\, \tilde M(\theta) \, d\theta, \quad C_2 = \frac{1}{4}\int^\pi_0 \cos\theta\sin^3\theta\, \tilde h(\theta)\,\tilde M(\theta)\, d\theta,$$
with the notations in Eq. \eqref{eq:notation_theta}.
The term $P_{\Omega^\perp}X_2$ has been studied in Ref. \cite{degond2015phase}. One observes that in the limit $j_{\rho M_\Omega}=c_1 \rho \Omega$ using Eq. \eqref{eq:consistency relationship} and, therefore, 
\begin{eqnarray*}
&&\hspace{-0.5cm}\nu\nabla_\omega \cdot \lp P_{\omega^\perp} \frac{k_0}{|j_{\rho M_\Omega}|} P_{\Omega^\perp}( \Delta_x j_{\rho M_\Omega}) \, \rho M_{\Omega}\rp \\
&&= \nu k_0 \nabla_\omega \cdot \lp P_{\omega^\perp}P_{\Omega^\perp}(\Delta_x(\rho\Omega)) \, M_{\Omega} \rp\\
&&= \nu k_0 \nabla_\omega \cdot \lp P_{\omega^\perp}P_{\Omega^\perp}\Delta_x(\rho \Omega) \rp\, M_{\Omega}+\,\nu k_0 \kappa\, P_{\omega^\perp}P_{\Omega^\perp}\Delta_x(\rho \Omega) \cdot  P_{\omega^\perp}\Omega\, M_\Omega\\
&&= -2 \nu k_0 (\omega\cdot P_{\Omega^\perp} (\Delta_x (\rho \Omega) )\, M_{\Omega}-\,\nu k_0\kappa (\omega \cdot P_{\Omega^\perp}(\Delta_x (\rho \Omega)) (\omega \cdot \Omega) \, M_{\Omega},
\end{eqnarray*}
where we used that $\nabla_\omega \cdot \lp P_{\omega^\perp} A \rp=-2 A\cdot \omega$, $\nabla_\omega (\omega \cdot A)=P_{\omega^\perp}A$ and $P_{\omega^\perp}A \cdot P_{\omega^\perp}B= A\cdot B - (\omega\cdot A) (\omega \cdot B)$ for any vectors $A,B\in \mathbb{R}^3$  (see Ref. \cite{frouvelle2012continuum}). With this expression we have that
\begin{align*}
P_{\Omega^\perp}X_2 =& -2\nu k_0 P_{\Omega^\perp}\lp\int_{\mathbb{S}^2} (\omega \otimes \omega) \,M_{\Omega}\, h \, d\omega\rp P_{\Omega^\perp} \Delta_x(\rho\Omega) \\
&-\nu k_0 \kappa\int_{\mathbb{S}^2} (\omega \cdot P_{\Omega^\perp}(\Delta_x (\rho \Omega)) (\omega \cdot \Omega) \, M_{\Omega}\, h\, P_{\Omega^\perp}(\omega)\,d\omega\\
=& -2 \nu k_0  C_0 \,P_{\Omega^\perp} \Delta_x (\rho \Omega)\\
&-\frac{\nu k_0 \kappa}{2}\int_{\mathcal{S}}\int^\pi_0
\big[(\cos\theta\, \Omega+ \sin\theta\, w) \cdot P_{\Omega^\perp}\Delta_x(\rho \Omega)\big] \cos\theta\, \tilde M(\theta)\, \tilde h(\theta)\,  \sin\theta\, w \, \sin \theta\, d\theta dw\\
=& -2\nu k_0  C_0\, P_{\Omega^\perp} \Delta_x (\rho \Omega)\\
&-\frac{\nu k_0\kappa}{2} \lp \int^\pi_0 \sin^3\theta \cos\theta \tilde M(\theta) \, \tilde h(\theta) \, d\theta\rp \lp \int_{\mathcal{S}} w\otimes w\, dw \rp P_{\Omega^\perp}\Delta_x(\rho\Omega) \\
=& -\nu k_0 \lp C_2\kappa+2C_0  \rp P_{\Omega^\perp} \Delta_x (\rho \Omega),
\end{align*}
where in the second equality we used the change of variable \eqref{eq:change of variables}, as well as, 
\begin{equation} \label{eq:usual_integral}
P_{\Omega^\perp}\int_{\mathbb{S}^2} (\omega \otimes \omega) \, h\, M_{\Omega}\, d\omega= \frac{1}{4}\int^\pi_0 \sin^3\theta\, \tilde h(\theta)\, \tilde M(\theta) \,   d\theta P_{\Omega^\perp}=: C_0 P_{\Omega^\perp},
\end{equation}
(this formula is a consequence of Eqs. \eqref{eq:decomposition_integral}-\eqref{eq:tensorial_w}); in the third equality, the odd integrands in $w$ vanish; and in the last equality we used Eq. \eqref{eq:tensorial_w}.

\medskip
Now, the terms $X_3$ and $X_4$ correspond to the coupling terms.
Firstly, for $X_3$ we have that
	\begin{align*}
		P_{\Omega^\perp}X_3
		 =& P_{\Omega^\perp}\int_{\mathbb{S}^2} \left[(v\cdot \nabla_x) \rho + \kappa \rho\, \omega \cdot ((v \cdot\nabla_x)  \Omega) \right] h\,M_\Omega\,\omega\, d\omega\\
		=&\kappa\rho\, P_{\Omega^\perp} \lp \int_{\mathbb{S}^2} (\omega\otimes \omega)\, h\, M_\Omega\, d\omega \rp (v\cdot\nabla_x) \Omega\\
		=& \kappa\, \rho \, C_0\, (v\cdot\nabla_x)\Omega,
	\end{align*}
where in the second equality the term $(v\cdot \nabla_x) \rho$ vanishes since 
$$P_{\Omega^\perp}\int_{\mathbb{S}^2}h \, M_{\Omega}\, \omega\, d\omega =\frac{1}{2}\int^\pi_0 \tilde h(\theta)\, \tilde M(\theta)\sin\theta \, d\theta\, \int_{\mathcal{S}}w\, dw =0, $$
and in the last equality  we used that $P_{\Omega^\perp} (v\cdot \nabla_x) \Omega= (v\cdot \nabla_x) \Omega$, as well as, Eq. \eqref{eq:usual_integral}.

Finally, to compute $X_4$ we denote by $B:= \lambda S(v) +A(v)$. Then we have that
\begin{align*}
\nabla_\omega \cdot (P_{\omega^\perp}(B\omega)\,  \rho M_\Omega)&= \nabla_\omega \cdot\lp P_{\omega^\perp}B\omega \rp \rho M_{\Omega}+ (P_{\omega^\perp}B\omega) \cdot \nabla_\omega(\rho M_\Omega)\\
&= B :(\mbox{Id}-3 \omega \otimes \omega)\, \rho M_\Omega+ \kappa \rho M_\Omega \left[ (\omega \cdot B^T \Omega) - (\omega \cdot B \omega) (\omega \cdot \Omega) \right],
\end{align*}
where we used that $\nabla_\omega M_\Omega= \kappa P_{\omega^\perp}\Omega M_{\Omega}$ and that $\nabla_\omega \cdot \lp P_{\omega^\perp}B\omega \rp= B: (\mbox{Id}-3\omega\otimes \omega)$ for any matrix $B$ independent of $\omega$. The notation $B:C$ indicates the contractions of the two matrices $B=(B)_{ij}$, $C=(C)_{ij}$, i.e., $B:C= \sum_{i,j=1,2,3}B_{ij}C_{ij}= \mbox{trace}(B^T C)$ (see Ref. \cite[Ap. A.2]{frouvelle2012continuum}). 
In this way we can decompose $X_4$ into $X_4= X_{41} + X_{42} + X_{43}$
with
\begin{align*}
X_{41}&= \rho\int_{\mathbb{S}^2} \big( B :(\mbox{Id}-3 \omega \otimes \omega) \big) \,  M_\Omega\, h\, \omega \,d\omega,\\
X_{42} &= \kappa \rho \int_{\mathbb{S}^2}  (\omega \cdot B^T \Omega)\, M_\Omega\, h\, \omega \, d\omega,\\
X_{43} &= -\kappa \rho \int_{\mathbb{S}^2} (\omega \cdot B \omega) (\omega \cdot \Omega)\, M_\Omega\, h\, \omega \, d\omega.
\end{align*}

To compute the term $X_{41}$, notice that, if $C$ is an antisymmetric matrix, then $C: (\mbox{Id}- 3\omega\otimes\omega)=0$ (since the second matrix is symmetric), therefore
$$B:(\mbox{Id}- 3\omega\otimes\omega)= \lambda\, S(v):(\mbox{Id}- 3\omega\otimes\omega).$$ 
In the following computation, in the second equality we use the change of variables \eqref{eq:change of variables}; in the third equality the odd terms in $w$ vanish from the integral; in the fourth equality we use that $S: (w\otimes \Omega + \Omega \otimes w)= 2 w \cdot S\Omega$ (since $S$ is symmetric); and the last equality is consequence of Eq. \eqref{eq:tensorial_w}:
\begin{align*}
P_{\Omega^\perp}X_{41} &= \rho\int_{\mathbb{S}^2} \big( B: (\mbox{Id}- 3\omega\otimes\omega) \big)\, M_{\Omega}\,h \,P_{\Omega^\perp}(\omega)\, d\omega \\
&= \frac{\lambda}{2} \, \rho\int_{\mathcal{S}}\int^\pi_0 \big[ 
 S(v): \big(\mbox{Id}-3(\cos\theta\, \Omega + \sin\theta\, w) \otimes (\cos\theta\,\Omega + \sin \theta\, w) \big)\big]\\
& \qquad\qquad\qquad \tilde M(\theta)\, \tilde h(\theta) \sin\theta w \sin \theta \, d\theta dw \\
&=-\frac{3\lambda}{2}\,\rho\lp \int^\pi_0\sin^3\theta\cos\theta \tilde M(\theta)\, \tilde h(\theta) \,d\theta\rp \lp \int_{\mathcal{S}} \big[ S(v): (w\otimes \Omega + \Omega \otimes w) \big] w \, dw\rp\\
&= -12\lambda\, C_2\rho\lp \int_{\mathcal{S}} w\otimes w \, dw \rp  S(v)\Omega\\
&= -6 \lambda\, C_2\,\rho\,   P_{\Omega^\perp} S(v)\Omega.
\end{align*}

For the term $X_{42}$ it is immediate to obtain
$$ P_{\Omega^\perp}X_{42} =C_0 \kappa\rho P_{\Omega^\perp}B^T \Omega= C_0\kappa\rho P_{\Omega^\perp}(\lambda S(v) + A(v)) \Omega,$$
proceeding analogously as in previous computations (remember Eq. \eqref{eq:usual_integral}). The term $X_{43}$ is computed similarly as for $X_{41}$: 
\begin{align*}
P_{\Omega^\perp} X_{43}&=-\kappa \rho \, P_{\Omega^\perp}\int_{\mathbb{S}^2}(\omega\cdot B \omega) \, (\omega\cdot \Omega) \, M_\Omega\, h\, \omega \, d\omega\\
&= -\kappa \rho\int_{\mathbb{S}^2} (\omega \cdot B\omega)\, (\omega\cdot \Omega)  \, M_\Omega\, h\, P_{\Omega^\perp}(\omega) \, d\omega\\
&= -\frac{\kappa}{2} \rho\int_{\mathcal{S}}\int^\pi_0\big[ (\cos\theta\, \Omega+\sin\theta\, w)\cdot B(\cos\theta\, \Omega+\sin\theta\, w) \big]\cos\theta\, \tilde M(\theta)\, \tilde h(\theta) \sin\theta\, w\sin\theta\, d\theta dw\\
&= -\kappa\, C_3\, \rho\int_{\mathcal{S}} (w\cdot (B+B^T) \Omega) w\, dw\\
&= -\kappa\, C_3\, \rho\int_{\mathcal{S}} (w \otimes w)  \, dw \,(B+B^T) \Omega\\
&=-\kappa \lambda\, C_3\, \rho\,P_{\Omega^\perp} S(v) \Omega,
\end{align*}
where in the last equality we substituted $(B+B^T)/2= \lambda S$; and where
$$C_3:=\frac{1}{2}\int^\pi_0 \sin^3\theta\cos^2\theta \tilde M(\theta) \tilde h(\theta)\, d\theta.$$

Grouping terms we conclude:
$$P_{\Omega^\perp} X_4 = \rho \left[ \kappa C_0 P_{\Omega^\perp}A(v)\Omega + \lambda\lp\kappa C_0 -6 C_2-\kappa C_3  \rp \right] P_{\Omega^\perp} S(v)\Omega $$

\medskip

Finally, putting all the terms together, we obtain
\begin{eqnarray*}
0=P_{\Omega^\perp}X &=& P_{\Omega^\perp}(X_1+X_2+X_3+X_4)\\
&=& \rho(C_0\kappa\,\partial_t \Omega + a\kappa C_2 (\Omega \cdot \nabla_x) \Omega) + aC_0 P_{\Omega^\perp} \nabla_x\rho\\
&&-\nu k_0 \lp C_2\kappa+2C_0  \rp P_{\Omega^\perp} \Delta_x (\rho \Omega)\\
&& +\kappa\, \rho \, C_0\, (v\cdot\nabla_x)\Omega\\
&& +\rho \left[ \kappa C_0 P_{\Omega^\perp}A(v)\Omega + \lambda\lp\kappa C_0 -6 C_2-\kappa C_3  \rp \right] P_{\Omega^\perp} S(v)\Omega.
\end{eqnarray*}
Dividing the previous expression by $\kappa C_0$ we obtain Eq. \eqref{eq:hydro2} for $\Omega$ with 
$$c_2= \frac{C_2}{C_0}, \quad c_3=\frac{C_3}{C_0} .$$

 \bigskip

To conclude the theorem, we are left with computing the limit for Stokes equation  \eqref{eq:scaledsys_v}. For this, we just need to compute the  limit of the right hand side term, which in the limit $\eps \to 0$ corresponds to 
$$
- b\nabla_x\cdot\left(\int_{\mathbb{S}^{2}}{\lp \omega\otimes\omega-\frac{1}{3}\Id \rp \rho M_{\Omega}(\omega)\,d\omega}\right) .
$$
We compute next the value of the integral:
\begin{eqnarray*}
&& \hspace{-1cm}\rho\int_{\mathbb{S}^{2}}\lp \omega\otimes\omega-\frac{1}{3}\Id\rp M_{\Omega}\,d\omega \\
&=& \frac{\rho}{2} \int_{\mathcal{S}}\int^\pi_0 \left[(\cos\theta\, \Omega +\sin\theta \, w) \otimes (\cos\theta\, \Omega +\sin\theta\, w) - \frac{1}{3}\Id\right] \tilde M(\theta) \sin\theta\, d\theta dw\\
&=&\frac{\rho}{2} \lp\int^\pi_0 \cos^2\theta \tilde M(\theta) \sin\theta\,d\theta \rp \Omega \otimes \Omega +\, \rho \lp\int^\pi_0 \sin^3\theta \tilde M(\theta)\, d\theta \rp \lp\int_{\mathcal{S}} w\otimes w\, dw \rp-\frac{1}{3}\rho\Id\\
&=& \frac{\rho}{2} \lp \int^\pi_0 \cos^2\theta \tilde M(\theta) \sin\theta \, d\theta \rp \Omega \otimes \Omega +\, \rho \lp \int^\pi_0 \sin^3\theta \tilde M(\theta) \, d\theta \rp \frac{1}{2} (\mbox{Id} -\Omega \otimes \Omega)-\frac{1}{3}\rho\Id\\
&=& \rho\left[c_4 \,\lp\Omega \otimes \Omega-\frac{1}{3}\Id\rp+ c_5\mbox{Id} \right],
\end{eqnarray*}
where in the third equality we have disregarded the odd terms in $w$; and where
$$c_4=\frac{1}{2}\int^\pi_0 \sin\theta \tilde M(\theta) \lp \cos^2\theta -\frac{1}{2}\sin^2\theta \rp\, d\theta,  \quad c_5=\frac{1}{4}\int^\pi_0 \sin^3\theta \tilde M(\theta) \, d\theta+\frac{1}{6}(c_4-1).$$
 A computation shows that $c_5=0$. This implies that
$$
- b\nabla_x\cdot\left(\int_{\mathbb{S}^{2}}{\lp \omega\otimes\omega -\frac{1}{3}\Id \rp \rho M_{\Omega}(\omega)\,d\omega}\right)  =  -b \,c_4 \nabla_x \cdot \left[ \rho \lp \Omega \otimes \Omega-\frac{1}{3}\Id\rp \right].
$$

\end{proof}

\section{Linearised stability analysis of the SOH-Stokes system}

\label{sec:stability}

In this section, we investigate the linearised stability of the SOH-Stokes system \eqref{eq:macro_SOH_Stokes}. We linearize the SOH-Stokes system about constant (space-independent) functions $\rho$, $\Omega$, $v$, $p$ and study the stability of the resulting linear system. 
The main result of this section is that the SOH-Stokes model exhibits unstable modes for both pushers $(b>0)$ and pullers $(b<0)$. Since the SOH model describes aligned states (as the particle distribution function is non-isotropic, given by a von Mises distribution with non-zero parameter $\kappa$), this corresponds to analyzing the stability of the suspension near an aligned state. A previous analysis performed in \cite{saintillan2} in the case of nematic interactions (see also \cite{hohenegger2010stability}) has shown that both pushers and pullers are unstable to perturbations of an aligned state. We show that this instability still prevails for both pushers and pullers interacting though polar alignment. However, we show that pullers can be stable if they are slender rods $(\lambda=1)$. We will also see that the unstable modes for pushers and pullers are not the same. In the case of pullers, these are transverse modes (the perturbation to $\Omega$ is normal to the wave-vector) propagating along the unperturbed orientation vector $\Omega$. For pushers, these are longitudinal modes propagating transversely to the unperturbed orientation vector $\Omega$. The former have vanishing density perturbation while the latter have non-trivial density perturbation. For both pushers and pullers, the instability only develops at small values of $|k|$ (i.e. for large wavelengths) and has maximal growth rate at $k=0$. Therefore, we can expect that the typical spatial extension of the instability patterns will be set up by the system size.

 Here we assume that $a=1$ to simplify the analysis. Let 
$$\rho=\rho_0, \, \Omega = \Omega_0, \, v= v_0, \, p=p_0,$$
 be a uniform steady state for the SOH-Stokes system with $|\Omega_0|=1$. We expand it with a small perturbation parameter $\tau$:
$$\rho=\rho_0+\tau \rho_1(x,t), \, \Omega=\Omega_0+\tau \Omega_1(x,t),\, v=v_0 +\tau v_1(x,t),\, p= p_0+\tau p_1(x,t).$$
Dropping the higher order terms $\mathcal{O}(\tau^2)$ and using $(\rho, \Omega, v, p)$ to  represent the first order perturbation (rather than $(\rho_1, \Omega_1, v_1, p_1)$) we obtain the linearised system:
\begin{subequations} \label{eq:linearised_system}
\begin{numcases}{}
			\Omega_0\cdot \Omega=0, \label{eq:linearised_1}\\
			\partial_t \rho+ \rho_0\nabla_x\cdot (c_1\Omega +v) + \big( (c_1\Omega_0+v_0) \cdot \nabla_x\big) \rho=0,\\
\rho_0 \partial_t \Omega + \rho_0\big( (c_2\Omega_0+v_0)\cdot \nabla_x\big) \Omega +\frac{1}{\kappa}P_{\Omega_0^\perp}\nabla_x \rho\nonumber\\
\qquad\qquad= \gamma \rho_0 P_{\Omega_0^\perp}\Delta_x\Omega + \rho_0 P_{\Omega_0^\perp} \big( \tilde \lambda S(v)+A(v) \big) \Omega_0, \label{eq:linearised_Omega}\\
-\Delta_x v+ \nabla_x \tilde p = -\tilde b \rho_0 \big( (\Omega_0\cdot\nabla_x) \Omega+ (\nabla_x\cdot \Omega) \Omega_0 \big)-\tilde b \lp\Omega_0 \otimes \Omega_0 \rp \nabla_x \rho, \quad \label{eq:linearised_v}\\
\nabla_x\cdot v=0, 
\end{numcases}
\end{subequations}
where $\tilde b= b c_4$ and
$$\tilde p= p -\frac{\tilde b}{3}\rho.$$
 The first equation is consequence of $|\Omega_0|=1$.
To deduce the first term in the right hand side of Eq. \eqref{eq:linearised_Omega} we used that $P_{\Omega_0^\perp} \Omega_0 \Delta_x \rho=0$. Finally, to obtain the last term in Eq. \eqref{eq:linearised_v} we used that
$$\mathcal{Q} = c_4 \lp \Omega_0\otimes \Omega_0 - \frac{1}{3}\Id \rp + \tau c_4 (\Omega_0\otimes \Omega+\Omega\otimes \Omega_0) + \mathcal{O}(\tau^2).$$

\medskip
The main result of this section is the following:
\begin{theorem}[Linear stability analysis]
\label{lem:plane_wave_solutions}
There exists a non-trivial plane wave solution for the linearised system \eqref{eq:linearised_system} of the form
\begin{equation} \label{eq:plane_wave}
(\rho, \Omega, v, \tilde p) = (\bar \rho, \bar \Omega, \bar v, \bar p) e^{i(k\cdot x-\alpha t)},
\end{equation}
(where $(\bar\rho, \bar\Omega,\bar v, \bar p)=(\bar\rho, \bar\Omega,\bar v, \bar p)(\alpha,k)$ are complex-valued functions, with $k\in \R$, $\alpha\in \mathbb{C}$), if and only if, either $k=k_0\Omega_0$ for some $k_0\in \R$, $k_0 \neq 0$ or $P_{\Omega_0^\perp}k\neq 0$, as detailed next.  Denote $ k_0, \bar k, k^\perp, U_0, V_0$ by
$$k_0:= k\cdot \Omega_0,\quad \bar k:= k\cdot \bar \Omega,\quad k^\perp:= P_{\Omega_0^\perp}k,$$
$$U_0:= c_1\Omega_0+v_0,\quad V_0:= c_2\Omega_0+v_0.$$
\begin{itemize}
\item[\mbox{}] \textbf{Case A:} $k=k_0\Omega_0$ for $k_0\in\R$, $k_0\neq 0$.\\
 In this case $\alpha$ can only have two possible values:
\begin{itemize}
 \item[(a)] either $\alpha= (c_1+v_0\cdot \Omega_0)k_0$, and then $\bar \Omega=0$, $\bar\rho$ is arbitrary, $\bar p= - \tilde b \bar\rho$, $\bar v=0$;
\item[(b)] or
$$\alpha= (c_1+v_0\cdot\Omega_0)k_0+i\lp \frac{\tilde \lambda -1}{2}\tilde b \rho_0- \gamma k_0^2\rp,$$ 
and therefore it is stable $(\mbox{Im}(\alpha)\leq 0)$ if
$$|k|^2 \geq \frac{1}{2\gamma}\rho_0 \tilde b(\tilde \lambda -1).$$
In this case $\bar \rho=0$, $\bar \Omega$ is an arbitrary unit vector orthogonal to $\Omega_0$, $\bar p=0$, 
$$\bar v=-i \tilde b \frac{\rho_0}{k_0}\bar\Omega.$$
If $\tilde b<0$ (puller case), the modes are unstable in the range
\be \label{eq:unstable_range_A}
|k|^2\in \left[0, \frac{\rho_0 \tilde b (\tilde \lambda-1)}{2\gamma} \right].
\ee
The supremum of $\mbox{Im}(\alpha)$ in this range is
\begin{equation} \label{eq:supremumA}
\frac{\rho_0 \tilde b (\tilde \lambda-1)}{2\gamma},
\end{equation}
and corresponds to the limit of $\mbox{Im}(\alpha)$ when $k\to 0$.

\end{itemize}
\item[\mbox{}] \textbf{Case B:} $k^\perp=P_{\Omega_0^\perp}k\neq 0$.\\
 In this case  $\bar \Omega$ is of the form
\begin{equation} \label{eq:form_Omega_bar}
\bar \Omega= \eta \frac{k^\perp}{|k^\perp|},
\end{equation}
with $\eta=\pm 1$ and $(\alpha, k)$ are linked by the following dispersion relation,
 $$D_\eta(\alpha, k)=0,$$
where
\beqarl
D_\eta (\alpha, k)&=&  \bar k\Bigg\{\frac{\tilde b \rho_0}{2|k|^2}\Bigg[ \lp -4\tilde \lambda\frac{k_0^2}{|k|^2}+\tilde \lambda +1 \rp
\lp-\alpha+U_0\cdot k \rp \nonumber\\
&&\qquad\quad-c_1 k_0\lp -2\tilde \lambda\frac{k_0^2}{|k|^2}+\tilde \lambda+1 \rp \Bigg]+\frac{i}{\kappa} c_1 \Bigg\}(|k|^2-k_0^2)^{1/2} \nonumber\\
&& -\eta (-\alpha+U_0\cdot k) \left[i(-\alpha+V_0\cdot k) - \frac{(\tilde \lambda -1)\tilde b \rho_0}{2}\frac{k_0^2}{|k|^2}+\gamma|k|^2\right]\!.\qquad \label{eq:dispersion_relation}
\eeqarl

\medskip

In the particular case  where $k_0=0$, the dispersion relation simplifies to 
\be \label{eq:simplified_dispersion_relation}
\tilde D(\alpha, k)= \rho_0\frac{\tilde b}{2}(\tilde \lambda+1)(-\alpha+v_0\cdot k)+\frac{i}{\kappa} c_1|k|^2-(-\alpha+v_0\cdot k)[i (-\alpha+v_0\cdot k)+\gamma |k|^2]=0.
\ee
The corresponding modes are stable $(\mbox{Im}(\alpha)\leq 0$) if
\be \label{eq:stability condition case B}
|k|^2\geq \frac{1}{2\gamma}\rho_0 \tilde b (\tilde \lambda +1),
\ee
and the perturbation is given by 
\be
\bar p = 0,\quad
\bar \rho = \eta c_1\frac{ \rho_0 |k|}{\alpha-v_0\cdot k},\quad
\bar v =-i\eta\tilde b\frac{\rho_0}{|k|}\Omega_0.
\ee
If $\tilde b >0$ (pusher case), the modes are unstable in the range
$$|k|^2 \in \left[0, \frac{\rho_0\tilde b (\tilde \lambda+1)}{2\gamma}\right].$$
The supremum of $\mbox{Im}(\alpha)$ in this range is
$$\frac{\rho_0\tilde b (\tilde \lambda+1)}{2},$$
and corresponds to the limit $k\to 0$. 

\end{itemize}
\end{theorem}

\begin{remark}
Notice that \emph{Case (B)} when $k_0=0$  corresponds to $\bar v \perp \bar \Omega$, while \emph{Case (A) (b)} corresponds to $\bar v \parallel \bar \Omega$. This is the signature that these two cases correspond to different modes.
\end{remark}

\begin{remark}[Interpretation of the linear stability analysis, Th. \ref{lem:plane_wave_solutions}]
\label{rem:stability_analysis}
\mbox{}\\
\textbf{Case (A) (a)}:
 This case corresponds to the simple propagation of a density perturbation along $\Omega_0$ at speed 
	$$\frac{\alpha}{k_0}= c_1 + v_0\Omega_0,$$
 with no perturbation of the orientation since $\bar \Omega=0$.
 \medskip
 
\noindent\textbf{Case (A)(b)}: Notice that $\tilde \lambda -1 \in [-2,0]$, since $\tilde \lambda \in [-1,1]$. Therefore, if $\tilde b>0$ (pushers), the mode is stable and if $\tilde b<0$ (pullers) the mode  is unstable for small values of $|k|$. In this last case, the coupling with Stokes equation destabilizes the model given that in the SOH model alone all modes are stable, see Ref. \cite{degond2008continuum}. Notice that the diffusion term helps to stabilize the modes by damping them,  but since it involves a second order derivative, the damping is proportional to $|k|^2$ and is very small for small values of $|k|$ but dominates for large values of $|k|$. Therefore, for large values of $|k|$, the diffusion damping is enough to compensate the instability due to the coupling with the Stokes equation,  which is independent of $|k|$. This is why the model is stable for large values of $|k|$. However, for small values of $|k|$, the diffusion damping is not strong enough and the instability of the Stokes coupling is predominant. Consequently, small $|k|$-modes (large wavelength) are unstable.
Moreover, the supremum  of $\mbox{Im}(\alpha)$ corresponds to the limit $k\to 0$, which means that the typical spatial extension of the fastest growing unstable mode will be of the size of the system.

\medskip
\noindent\textbf{Case (B)}: We analyse the particular case $k_0=0$. We observe analogous phenomena as in Case (A)(b) but reversing the roles of pullers and pushers since $\tilde \lambda+1\geq 0$:
if $\tilde b<0$ (pullers), the constant solution is always stable but in the 
 pusher case $(\tilde b>0)$, the coupling with Stokes equation destabilizes the mode. 
The supremum value of $\mbox{Im}(\alpha)$  also corresponds to the limit $k\to 0$.

\end{remark}

\begin{remark}
Figures \ref{Fig:instability}-\ref{Fig:pusher_instability} provide a schematical explanation of the instability mechanisms. Fig. \ref{Fig:instability} depicts  the perturbation velocity field generated  by pushers and pullers. Fig. \ref{Fig:puller_instability} and \ref{Fig:pusher_instability} provide  a description of the instability mechanisms for pullers and pushers respectively.
\end{remark}

\begin{figure*}[t!]
\begin{subfigure}[b]{0.45\textwidth}
\includegraphics[width=\textwidth]{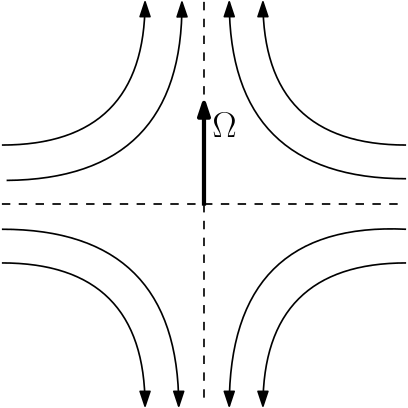}
\caption{ }
\label{Fig:inst1}
\end{subfigure}
\hfill
\begin{subfigure}[b]{0.45\textwidth}
\includegraphics[width=\textwidth]{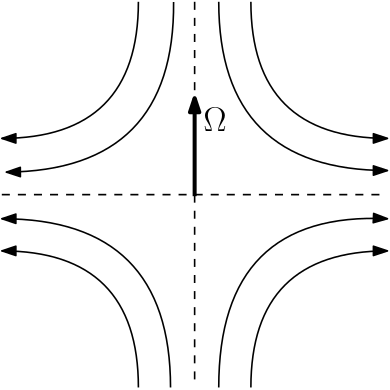}
\caption{ }
\label{Fig:inst2}
\end{subfigure}
\caption{\ref{Fig:inst1} Flow field generated by a pusher. \ref{Fig:inst2} Flow field generated by a puller. These flow fields generate perturbations to  the background velocity field and in some cases,  can provide the necessary positive feedback mechanism to trigger an instability. Instability mechanisms are different for pushers and pullers. } 
\label{Fig:instability}
\end{figure*}

\begin{figure*}[t!]
\begin{subfigure}[b]{0.14\textwidth}
\includegraphics[width=\textwidth]{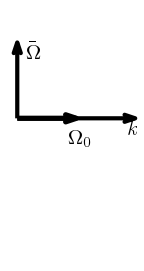}
\caption{ }
\label{Fig:puller_inst1}
\end{subfigure}
\hfill
\begin{subfigure}[b]{0.8\textwidth}
\includegraphics[width=\textwidth]{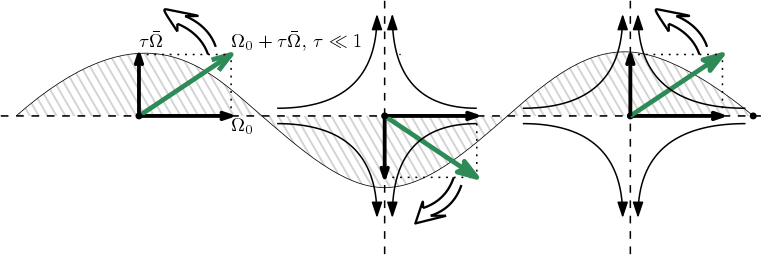}
\caption{ }
\label{Fig:puller_inst2}
\end{subfigure}
\caption{Puller unstable mode: \ref{Fig:puller_inst1} geometric configuration. The puller unstable mode is a transverse mode $(\bar \Omega \perp k)$ propagating parallel to the unperturbed mean orientation $\Omega_0$. \ref{Fig:puller_inst2} Schematics of the instability mechanism. There is no density perturbation involved. The instability is due to a reinforcement of the mis-alignment between  the swimmer mean orientation  $\Omega= \Omega_0+\tau \bar \Omega$, $(\tau \ll 1)$ and the unperturbed orientation  $\Omega_0$. This reinforcement results from  the torque applied to  a given  swimmer by  the velocity  perturbation generated  by the neighbouring swimmers ahead and behind it. This torque is materialised in the picture by the double arrows.}
\label{Fig:puller_instability}
\end{figure*}

\begin{figure*}[t!]
\begin{subfigure}[b]{0.14\textwidth}
\includegraphics[width=\textwidth]{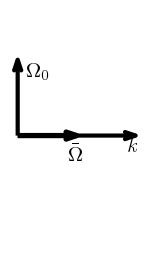}
\caption{ }
\label{Fig:pusher_inst1}
\end{subfigure}
\hfill
\begin{subfigure}[b]{0.8\textwidth}
\includegraphics[width=\textwidth]{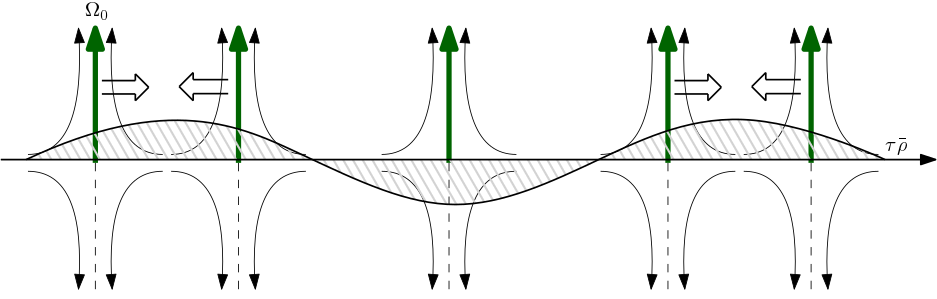}
\caption{ }
\label{Fig:pusher_inst2}
\end{subfigure}
\caption{Pusher unstable mode: \ref{Fig:pusher_inst1} geometric configuration. The pusher unstable mode is a longitudinal mode. The perturbation  $\bar \Omega$ is parallel to the propagation direction $(\bar \Omega \parallel k)$ and both are perpendicular to the unperturbed  mean orientation $\Omega_0$. \ref{Fig:pusher_inst2} Schematics of the instability mechanism. Due to the configuration of the perturbation velocity field that they generate, swimmers are attracted	 by regions of higher swimmer density,  thereby amplifying density perturbations.}
\label{Fig:pusher_instability}
\end{figure*}

\begin{proof}[Proof of Th. \ref{lem:plane_wave_solutions}]
Substituting the plane-wave solution \eqref{eq:plane_wave} in the linearised system \eqref{eq:linearised_system} we obtain
\begin{subequations} \label{eq:fourier}
\begin{numcases}{}
			\Omega_0\cdot \bar\Omega=0, \label{eq:fourier_1}\\
			-\alpha\bar\rho+ \rho_0 c_1\bar k + \bar\rho  U_0 \cdot k =0, \label{eq:fourier 2}\\
-i\alpha\rho_0 \bar\Omega + i\rho_0\big( V_0\cdot k\big) \bar \Omega +\frac{i}{\kappa}\bar\rho k^\perp \nonumber\\
\qquad \qquad=-\gamma |k|^2 \rho_0  \bar \Omega + \frac{i}{2} \rho_0 \left[ (\bar v \cdot \Omega_0) (\tilde \lambda+1) k^\perp + k_0 (\tilde \lambda -1) P_{\Omega_0^\perp}\bar v \right],\label{eq:fourier_Omega}\\
|k|^2 \bar v+ik\bar p = -i\tilde b \rho_0 \big( k_0\bar \Omega+ \bar k \Omega_0 \big)-i \tilde b\,\bar\rho k_0 \Omega_0, \quad \label{eq:fourier_v}\\
k\cdot \bar v=0,\label{eq:fourier_divergence} 
\end{numcases}
\end{subequations}
where in \eqref{eq:fourier 2} we used that $k\cdot \bar v=0$ thanks to \eqref{eq:fourier_divergence}; in \eqref{eq:fourier_Omega} we used that $P_{\Omega_0^\perp}\bar \Omega=\bar \Omega$ thanks to \eqref{eq:fourier_1}, as well as,   that $\nabla_x v = i k\otimes \bar v$ and therefore
$$S(v)= i \frac{k\otimes \bar v + \bar v \otimes k}{2}, \quad A(v)= i \frac{k\otimes \bar v - \bar v \otimes k}{2}, $$
and so
$$S(v)\Omega_0 = i\frac{(\bar v \cdot \Omega_0)k + k_0 \bar v}{2}, \quad A(v)\Omega_0 = i\frac{(\bar v \cdot \Omega_0)k - k_0 \bar v}{2}.$$

Now, we look for the existence of a non-trivial solution of system \eqref{eq:fourier}.  From the last two equations we deduce that
\begin{equation} \label{eq:for p}
 \bar p= -\tilde b \frac{k_0}{|k|^2}(2\rho_0 \bar k+\bar\rho k_0).
 \end{equation}
 We note that we can divide by $|k|^2$ since $k\neq 0$. Otherwise, if $k=0$, then in the case $\alpha \neq 0$ (which is the case of a non-trivial perturbation we are interested in), this implies that $\bar\rho=0$, $\bar \Omega=0$, i.e., the perturbation is null.
Next, we obtain an expression for $\bar v$ by decomposing it into $\bar v = P_{\Omega_0^\perp}\bar v + P_{\Omega_0}\bar v$, since this will be useful in the sequel. Doing the inner product of Eq. \eqref{eq:fourier_v}  with $\Omega_0$ and using Eq. \eqref{eq:fourier_1}, we obtain:
\begin{equation} \label{eq:for v parallel}
\bar v\cdot \Omega_0 = \frac{-i}{|k|^2} \lp  k_0 \bar p + \tilde b \rho_0 \bar k+\tilde b \bar \rho k_0\rp.
\end{equation}
Projecting now Eq. \eqref{eq:fourier_v} on the orthogonal to $\Omega_0$ we obtain
\begin{equation}\label{eq:for v orthogonal}
P_{\Omega_0^\perp} \bar v=-\frac{i}{|k|^2} \lp \bar p k^\perp + \tilde b\rho_0 k_0\bar \Omega\rp.
\end{equation}
We insert these expressions in \eqref{eq:fourier_Omega} to obtain:
\begin{eqnarray} 
&&\hspace{-1cm}\left[i\rho_0\lp -\alpha+V_0\cdot k \rp +\gamma|k|^2\rho_0 +\frac{(\tilde \lambda-1)\tilde b}{2}\frac{ \rho^2_0k_0^2}{|k|^2}\right]\bar\Omega\nonumber\\
 &=& \left[\frac{1}{2}\frac{\rho_0}{|k|^2}\lp 2\tilde \lambda k_0 \bar p+\tilde b (\tilde \lambda+1) ( \rho_0 \bar k + \bar \rho k_0)\rp-\frac{i}{\kappa}\bar\rho\right] k^\perp. 
\label{eq:for Omega bar}
\end{eqnarray}

Next, to study the solutions of this equation, we consider separately the cases $k^\perp=0$ and $k^\perp\neq 0$:

\noindent \paragraph{Case (A)} Suppose $k^\perp=0$, i.e. $k=k_0 \Omega_0$ with $k_0\neq 0$. We can distinguish two cases:
\begin{itemize}
\item[(a)] Suppose $\bar \Omega=0$, then Eq. \eqref{eq:fourier_1} and the fact that $\bar\rho\neq 0$ (otherwise the perturbation is null) give $\alpha=(c_1+v_0\cdot\Omega_0)k_0$. In this case one can check that $\bar \rho$ is arbitrary, $\bar p=-\tilde b \bar \rho$ and $\tilde v =0$.
\item[(b)] Suppose $\bar \Omega \neq 0$. Then, from Eq. \eqref{eq:for Omega bar}, it must hold (remember that $|k|^2=k^2_0)$) 
$$\alpha= (c_2 + v_0\cdot \Omega_0)k_0 + i\frac{(\tilde \lambda-1)\tilde b}{2}\rho_0 - i\gamma k_0^2.$$
The condition for stability is $\mbox{Im}(\omega) \leq 0$, i.e.,
$$|k|^2=|k_0|^2\geq \frac{1}{2\gamma}\rho_0 \tilde b (\tilde \lambda -1).$$ 
In this case one can check that $\bar\rho=0$, $\bar \Omega$ is arbitrary with $\bar \Omega, \Omega_0\neq 0$, $\bar p = 0$ and 
$$\bar v= -i\tilde b\frac{\rho_0}{k_0}\bar \Omega.$$
Moreover, for $\tilde b<0$, it is straightforward to see that the range for which $|k|$ is unstable is given by \eqref{eq:unstable_range_A} and the supremum of $\mbox{Im}(\alpha)$ is attained at \eqref{eq:supremumA} in the limit $k\to 0$.

\end{itemize}

\noindent \paragraph{Case (B)} Suppose that $k^\perp \neq 0$. 
The coefficient on the right-hand side of Eq. \eqref{eq:for Omega bar} is written (thanks to \eqref{eq:for p}) as
$$X:=\frac{\rho_0}{2|k|^2}\left[ 2\tilde \lambda k_0 \lp- \frac{\tilde b k_0}{|k|^2}(2\rho_0 \bar k+\bar \rho k_0) \rp +(\tilde \lambda+1) (\tilde b \rho_0 \bar k+\tilde b \bar\rho k_0) \right] - \frac{i}{\kappa}\bar\rho.$$
First we check that $X\neq 0$. Suppose that $\bar\rho= 0$. From Eq. \eqref{eq:fourier_1}, this implies that $\bar k=0$. So $X=0$ and we conclude that $\bar \Omega=0$. So the perturbation is null. Therefore it cannot be that $\bar\rho=0$, which implies $\mbox{Im}(X)\neq 0$, so that $X\neq 0$. Therefore, $\bar \Omega \neq 0$ and from \eqref{eq:for Omega bar}, it should be given by Eq. \eqref{eq:form_Omega_bar}. Then, from \eqref{eq:for Omega bar} again and the fact that $k^\perp= (|k|^2-k_0^2)^{1/2}$, the dispersion relation is given by:
\begin{eqnarray} 
&&\hspace{-1cm}\left\{\tilde b \frac{\rho_0}{2|k|^2}\left[ \rho_0 \bar k \lp -4\tilde \lambda \frac{k_0^2}{|k|^2}+\tilde\lambda+1\rp + \bar \rho k_0 \lp -\frac{2\tilde \lambda k_0^2}{|k|^2}+\tilde \lambda+1\rp\right]- \frac{i}{\kappa}\bar \rho \right\}(|k|^2-k_0^2)^{1/2}\qquad\nonumber\\
 && \qquad\qquad =\eta \left[i\rho_0 (-\alpha + V_0 \cdot k) - (\tilde \lambda-1) \frac{\tilde b \rho^2_0}{2}\frac{k_0^2}{|k|^2}+\gamma |k|^2 \rho_0 \right],
\label{eq:for Omega bar 2}
\end{eqnarray}
and from Eq. \eqref{eq:fourier 2} we have the relation
\be \label{eq:aux_rho_bar}
 (-\alpha + U_0\cdot k) \bar \rho + \rho_0 c_1 \bar k=0.\ee
We check that $-\alpha + U_0 \cdot k\neq 0$ by contradiction. Suppose that $-\alpha + U_0 \cdot k=0$, then, from the previous equation, we deduce that $\bar k = \bar \Omega \cdot k = \bar \Omega \cdot (k^\perp +k_0\Omega_0)=0$. From Eq. \eqref{eq:form_Omega_bar}, we get that $\bar \Omega \parallel k^\perp$, which implies that $|k^\perp|=0$. This contradicts our assumption that $k^\perp \neq 0$ and therefore, we conclude that, effectively, $-\alpha + U_0 \cdot k\neq 0$. So, multiplying Eq. \eqref{eq:for Omega bar 2} by $-\alpha + U_0 \cdot k\neq 0$ and using Eq. \eqref{eq:aux_rho_bar}, we get the dispersion relation in Eq.  \eqref{eq:dispersion_relation}.

\medskip
Now, to simplify the analysis we will restrict ourselves to the case where $k^\perp=k$, i.e. $k_0=k\cdot \Omega_0=0$. This implies, in particular, that $U_0\cdot k=V_0\cdot k= v_0\cdot k$, as well as,
$$\bar k = k\cdot \bar \Omega = k^\perp \cdot \bar \Omega=\eta |k^\perp|=\eta|k|.$$
With these considerations one can simplify the dispersion relation \eqref{eq:dispersion_relation} into
$$\tilde D(\alpha,k)=0,$$
where $\tilde D(\alpha,k)$ is given in 
 Eq. \eqref{eq:simplified_dispersion_relation}.

Using the variable
$X=\alpha - v_0\cdot k$ we can recast $\tilde D(\alpha,k)=0$ into:
$$X^2+iX\lp\gamma |k|^2-\frac{\rho_0\tilde b (\tilde \lambda+1)}{2}\rp- \frac{c_1|k|^2}{\kappa}=0,$$
after multiplying by $i$.
Now, changing variables $X=iY$, we have that $Y$ solves
\be \label{eq:Y}
P(Y):=Y^2+Y \lp\gamma|k|^2-\frac{\rho_0\tilde b (\tilde \lambda+1)}{2}\rp + \frac{c_1|k|^2}{\kappa}=0.
\ee
Stability in this case means $\mbox{Im}(X)<0$, i.e. $\mbox{Re}(Y)<0$. The polynomial $P$  has  real coefficients. There are two possibilities:
\begin{itemize}
	\item If  $P$ has positive discriminant, its two roots are real. In this case, to have stability we require them both to be negative, i.e. their product $\pi$ has to be positive and their sum $\sigma$ negative. The product is given by 
	$$\pi=\frac{c_1|k|^2}{\kappa} \geq 0,$$
	 and their sum is 
\be \label{eq:sum_roots}	
	\sigma=- \gamma|k|^2+ \frac{\rho_0\tilde b (\tilde \lambda+1)}{2} .
\ee
	So in this case the stability criteria corresponds to  $\sigma\leq 0$, which leads to Eq. \eqref{eq:stability condition case B}.
	\item If the polynomial $P$ has negative discriminant, the two roots are complex conjugate. Their real part is half their sum $\sigma$.
	 So again the stability criterion reduces to asking that $\sigma$ is negative, and we are left with the same stability criterion \eqref{eq:stability condition case B} as before. 
	
\end{itemize}

\bigskip
We suppose now that $\tilde b>0$, and we want to determine the supremum  on the instability range
$$|k|\in \left[ 0, \lp \frac{\rho\tilde b(\tilde \lambda+1)}{2\gamma}\rp^{1/2}\right],$$
and the corresponding value of $|k|$.
This supremum corresponds to $k_{max}=\mbox{argmax}\, \mbox{Re}(Y)= \mbox{argmax Im}(\alpha)$. We have seen that in the case where the  roots are real, they have the same sign. Therefore,  any root is less than the sum  $\sigma(|k|)$ given by \eqref{eq:sum_roots}.

The maximum value of $\sigma(|k|)$ is at $k=0$, i.e.,
$$\sigma(0)= \frac{\rho_0\tilde b (\tilde \lambda+1)}{2}.$$
One can easily check that $\sigma(0)$ is a root of $P$ for $|k|=0$ (the other root being 0) and therefore one of the roots is maximal at $|k|=0$.
\end{proof}

 \section{Extensions of the model}
 \label{sec:refinements}

 \subsection{Adding short-range repulsion}
 \label{sec:repulsion}
 The Vicsek-Stokes coupling \eqref{eq:IBM1}-\eqref{eq:IBM4} presented here can be extended towards different directions. Particularly, in regions where agents become highly packed, a repulsion force can be enforced between neighbouring particles to better account for volume exclusion. This can be easily done following Ref. \cite{degond2015macroscopic} where repulsion is introduced in the Vicsek model and coarse-grained into the Self-Organised Hydrodynamic model with Repulsion (SOHR). Particularly the individual based model corresponds to:
		\begin{empheq}[left=\empheqlbrace]{align}
			&dX_i = u_i dt= v(X_i,t) dt + a\omega_i dt -\mu (\nabla_x\Phi)(X_i,t) , \label{eq:IBM1_r}\\
			&d\omega_i = P_{\omega_i^{\perp}}\circ \Big[\nu \overline{\omega}_idt  -\xi (\nabla_x \Phi)(X_i,t)	 dt+ \sqrt{2D}\,dB_{t}^{i}\,\, +\big(\lambda S(v) + A(v)\big)\omega_i dt\Big], \label{eq:IBM2_r}\\
			&\bar\omega_i = \frac{J_i}{|J_i|} \text{ with } J_i = \sum_{k=1}^N  K\lp \frac{|X_i-X_{k}|}{R} \rp\omega_k, \label{eq:IBM2.1_r}\\
			&-\Delta_x v + \nabla_x p =  - \frac{b}{N}\sum_{i=1}^{N}\lp \omega_i\otimes\omega_i-\frac{1}{3}\Id\rp\nabla_x \delta_{X_i(t)},\label{eq:IBM3_r}\\ 
			&\nabla_x \cdot v = 0, \label{eq:IBM4_r}
		\end{empheq} 
with the same notations as for the system \eqref{eq:IBM1}--\eqref{eq:IBM4}, where $\mu, \xi>0$ and the repulsive potential $\Phi$ is defined as
$$\Phi(x,t)=\frac{1}{N}\sum_{k=1}^N \phi\lp \frac{|x- X_k(t)|}{r}\rp,$$
where $\phi=\phi(|x|)$ is a binary repulsion potential that only depends on the distance, and where $r>0$ is the typical repulsion range. We assume that $x\mapsto \phi(|x|)$ is smooth, as well as,
$$\phi\geq 0, \quad \int_{\R^3}\phi(|x|)\, dx<\infty,$$
which implies, in particular, that $\phi(|x|)\to 0$ as $|x|\to \infty$. 

The only differences between System \eqref{eq:IBM1_r}--\eqref{eq:IBM4_r} with the original Vicsek-Stokes system \eqref{eq:IBM1}--\eqref{eq:IBM4} are the addition of two new terms: the last term
 to the evolution of $X_i(t)$ in Eq. \eqref{eq:IBM1_r}, which expresses the repulsion force, and the second term 
in the evolution of $\omega_i(t)$ in Eq. \eqref{eq:IBM2_r}, which is a relaxation term of $\omega_i$ towards the force $\nabla_x\Phi(X_i(t),t)$. This terms models the fact that particles tend to actively align their directions of motion with the  force.

The presence of these new terms modifies the coarse-grained equations. To begin with, the mean-field equations correspond to (following Sec. \ref{sec:mean_field} and Ref. \cite{degond2015macroscopic}):
		\begin{equation}
			\begin{cases}
			&\partial_t f + \nabla_x\cdot(u_{(f,v)}f) \\
			&\qquad+ \nabla_{\omega}\cdot\Big( \big[P_{\omega^{\perp}} \left\{\nu \overline{\omega}_{f} -\xi \nabla_x \Phi_f(x,t)+\lp\lambda S(v)+ A(v)\rp\omega\right\}\big]f\Big) = D\Delta_{\omega}f, \label{eq:kinetic_eq_repulsion}\\
			&u_{(f,v)}(x,\omega,t)=v(x,t) + a\omega-\mu\nabla_x \Phi_f(x,t), \\
			&-\Delta_x v + \nabla_x p =  - b\nabla_x\cdot Q_{f},\\
			&\nabla_x \cdot v =0,
			\end{cases}
		\end{equation}
following the notations of Prop. \ref{prop:mean-field} and where
$$\Phi_f (x,t) = \int_{\R^3}\phi\left(\frac{|x-y|}{r} \right)\rho_f(y,t)\,dy.$$
The difference with respect to the mean-field system in \eqref{eq:kinetic_eq} is the extra term $\xi\nabla_x\Phi_f(x,t)$ in the equation for $f$ and the term $-\mu\nabla_x \Phi_f$ in the equation for the velocity $u_{(f,v)}$. 

To perform the macroscopic limit, we rescale the mean-field equations \eqref{eq:kinetic_eq_repulsion} analogously as  in Sec. \ref{sec:mean_field} adding the rescaling of $r=\eps \tilde r$ for $\tilde r>0$. Remember that the alignment interaction range $R$ is rescaled as $R=\sqrt{\eps} \tilde R$, therefore the alignment interaction range is larger than the repulsive range. Skipping the tildes, we obtain the rescaled system:
	\begin{equation}
			\begin{cases} \label{eq:rescaled-mf-repulsion}
			&\varepsilon\left[\partial_t f^{\varepsilon} + \nabla_x\cdot(u_{(f^{\varepsilon},v^\eps)}f^{\varepsilon})\right]\\
			& + \nabla_{\omega}\cdot\Big( \big[P_{\omega^{\perp}} \{ \nu\overline{\omega}_{f^{\varepsilon}}-\eps\xi \nabla_x \Phi_{f^\eps}(x,t)+\varepsilon\lp\lambda S(v^\eps)+A(v^\eps)\rp\omega\}\big]f^{\varepsilon}\Big) = D\Delta_{\omega}f^{\varepsilon}, \\
			&u_{(f^{\varepsilon},v^\eps)}(x,\omega,t)=v^\eps(x,t) + a\omega-\mu\nabla_x \Phi_{f^\eps}(x,t), \\
			& \bar \omega^\eps_{f}= \frac{J^\eps_{f}}{|J^\eps_{f}|}, \quad J^\eps_{f} = \int_{\mathbb{S}^2 \times \R^3} \omega K\lp\frac{|x-y|}{\sqrt{\eps}R} \rp f\, d\omega dy, \\
			& \Phi^\eps_f = \int_{\R^3}\phi \lp \frac{|x-y|}{\eps r}\rp \rho_{f}(y,t)\, dy,\\
			&-\Delta_x v^\eps + \nabla_x p^\eps =  - b\nabla_x\cdot G_{f^{\varepsilon}},\\
			&\nabla_x\cdot v^\eps=0.
			\end{cases}
		\end{equation}

From here we obtain the macroscopic equations as $\eps \to 0$:
\begin{theorem}[Macroscopic equations with volume exclusion]
\label{th:macro_volume_exclusion}
Consider the rescaled system \eqref{eq:rescaled-mf-repulsion}. When $\eps \to 0$, it holds (formally) that
$$(f^\eps, v^\eps, p^\eps)\to ( f = \rho M_\Omega, v, p),$$
where $\rho=\rho(x,t)\geq 0$ and $\Omega = \Omega (x,t)\in \mathbb{S}^{2}$ are the limits of the local density $\rho^\eps$ and 
the local mean orientation $\Omega_{f^\eps}$ in Eqs. \eqref{eq:density_local},\eqref{localaverageorient}, respectively.
Moreover, if the convergence is strong enough and $\Omega$, $\rho$, $v$ and $p$ are smooth enough, they satisfy the coupled system
\begin{empheq}[left=\empheqlbrace]{align}
		& \partial_t\rho + \nabla_x\cdot(\rho U)  = 0, \label{eq:continuity_equation_repulsion} \\
		 &\rho\partial_{t}\Omega + \rho (V\cdot\nabla_x)\Omega + P_{\Omega^{\perp}}\nabla_x p(\rho)= \gamma P_{\Omega^{\perp}}\Delta_x(\rho\Omega) +\rho P_{\Omega^{\perp}} \lp\tilde\lambda S(v) +A(v)\rp \Omega , \label{eq:hydro2_repulsion} \\
		&-\Delta_x v + \nabla_x p =  - b\nabla_x\cdot\lp \rho \mathcal{Q}(\Omega) \rp, \label{eq:hydro3_repulsion}\\ 
		&\nabla_x\cdot v=0 \label{eq:hydro4_repulsion}, 
	\end{empheq}
	where 
\begin{eqnarray*}
		&&	U = a c_1 \Omega +v -\mu \Phi_0\nabla_x \rho, \quad
		V = a c_2 \Omega + v - \mu \Phi_0\nabla_x\rho,\\ &&p(\rho)=\frac{a}{\kappa} \rho +\xi\mu\Phi_0\lp \frac{2}{\kappa} + c_2\rp\frac{\rho^2}{2}, \quad
		\mathcal{Q} = c_4\lp \Omega \otimes \Omega-\frac{1}{3}\emph{Id} \rp, \\
		&& \Phi_0=\int_{\R^3}\phi(x)\, dx,
\end{eqnarray*}
	and where the constants $c_1,\hdots, c_4$, $k_0$, $\gamma$ and $\lambda$ are given by Eqs. \eqref{eq:c1_th}--\eqref{eq:c4}, \eqref{eq:def_k0}, \eqref{eq:lambda0} and where $\kappa=\nu/D$.

\end{theorem} 
The proof of this theorem is direct from the proof of Th. \ref{th:hydro_limit} and the results in Ref. \cite{degond2015macroscopic}.

\begin{remark}[Discussion of the result.]
\label{rem:discussion_volume_exclusion}
 The repulsive force intensity is given by the parameter $\mu\Phi_0$. Observe that when $\mu\Phi_0=0$ we recover  the SOH-Stokes system \eqref{eq:continuity_equation}--\eqref{eq:hydro4}.
Notice that the presence of the repulsion modifies the velocity $U$ of $\rho$ in \eqref{eq:continuity_equation_repulsion} and the convective velocity $V$ of $\Omega$ in \eqref{eq:hydro2_repulsion} by adding a term $-\mu\Phi_0\nabla_x\rho$. This term in \eqref{eq:continuity_equation_repulsion} gives rise to a diffusion-type term for $\rho$ of the form $-\mu\Phi_0\nabla_x\cdot(\rho \nabla_x\rho)$, which resembles a porous-medium equation and that prevents the formation of high particle concentrations. In the case of the convective velocity of $\Omega$, this term indicates the tendency of particles to change their orientation towards regions of lower concentration. The other important difference is the presence of a non-linear term in the pressure $\nabla_x p(\rho)$ for $\Omega$ in \eqref{eq:hydro2_repulsion} which  increases the pressure effects, due to the repulsion forces, when the concentrations become high.
\end{remark}
 
 \subsection{Vicsek-Navier-Stokes coupling}

\label{sec:derivationIBM}

\subsubsection{The individual based model}
In a finite Reynolds number regime, fluid dynamics is described by the Navier-Stokes equations rather than by the Stokes equation. In this section, we propose a Vicsek-Navier-Stokes coupling also assuming finite particle inertia and derive the coarse-grained equations.  We will also see how the Vicsek-Stokes coupling in Eqs. \eqref{eq:IBM1}-\eqref{eq:IBM4} is obtained from this Vicsek-Navier-Stokes coupling by assuming a low Reynolds number regime and negligible particle inertia. We consider the following coupled system:
\begin{subequations}\label{IBM_full}
\begin{numcases}{}
\frac{dX_i}{dt} = u_i(t),\label{IBM_full_a}\\
m_i \frac{d u_i}{dt} =\eta( v(X_i,t) + a\omega_i(t)- u_i), \label{IBM_full_a2}\\
 d\omega_i = P_{\omega_i^\perp} \circ \big[\nu \bar\omega_idt   +\big(\lambda S(v)+A(v)\big)\omega_i dt\big]+ \sqrt{2D}dB_t^i, \quad\label{IBM_full_b}\\
m_i \frac{du_i}{dt}= F_i(t), \label{IBM_full_c}\\
			\bar\omega_i = \frac{J_i}{|J_i|} \text{ with } J_i = \sum_{k=1}^NK(|X_i-X_{k}|)\omega_k, \label{eq:IBM_full_average}\\
\rho_0(\partial_t v + (v\cdot\nabla_x) v) + \nabla_x p = \sigma\Delta_x v - \sum_{i=1}^N F_i\delta_{X_i(t)} \nonumber\\
\qquad\qquad \qquad\qquad \qquad\qquad- \rho_0 \beta  \frac{1}{N}\sum_{i=1}^N \lp\omega_i \otimes \omega_i- \frac{1}{3}\Id \rp\nabla_x\delta_{X_i(t)}, 
\label{IBM_full_d}\\
\nabla_x \cdot v = 0. \label{IBM_full_e}
\end{numcases}
\end{subequations}
Most of the terms have previously been explained for Eqs. \eqref{eq:IBM1}-\eqref{eq:IBM4} in Sec. \ref{sec:discussion}. The term $\rho_0$ is the density of the fluid and $\sigma>0$ its viscosity; $\eta$ is a friction coefficient; $m_i$ is the mass of agent $i$; and $F_i$ is the force generating its acceleration. Notice that in the present case the individuals' velocity $u_i$ relaxes towards $v(X_i,t)+a\omega_i(t)$, while in the Vicsek-Stokes coupling we considered directly the relaxed system.
The influence of the force of particle $i$ on the fluid is given by the term $F_i \delta_{X_i(t)}$ in Eq. \eqref{IBM_full_d} (this is an application of Newton's third law of action and reaction).

A sanity check of our model consists of ensuring that the momentum and the angular momentum are conserved by the dynamics, as expected in a closed system with no dissipation at the boundaries:

\begin{proposition}
\label{prop:conservation_momentum}
Suppose that in the system \eqref{IBM_full}  the domain has no boundaries and the solution vanishes at large distances, then the total momentum and angular momentum are conserved.
\end{proposition} 
The proof can be found in the Appendix.

\subsubsection{Dimensional analysis and simplifications}

Next we check the orders of magnitude of the coefficients in Eqs. \eqref{IBM_full} by a dimensional analysis. 
We assume that each agent has the same mass $m=m_i$. We  consider dimensionless variables $x'=x/x_0$, $t'=t/t_0$ such that $x_0/t_0=u_0$ is the typical speed of an agent. With this, we define the dimensionless quantities 
\begin{eqnarray*}
 &&v'=v/u_0,\, a'=a/u_0,\, \nu'=\nu t_0,\, D' = D t_0,\\
 && \eta'= \eta \frac{t_0}{m},\, F_i' = F_i\lp m\frac{u_0}{t_0} \rp^{-1},\,  p'= p \lp \frac{\sigma u_0}{x_0} \rp^{-1}. 
\end{eqnarray*}
Now, we assume that the range of interaction of $K$ is given by $R$, so we can write 
$$K(x) = \tilde K\lp\frac{x}{R} \rp.$$
We introduce the dimensionless variable $R'=R/x_0$ so that $K' =\tilde K(|x'-y'|/R')$. 

Changing variables and expressing the system \eqref{IBM_full} in the prime variables we obtain, after skipping the primes, the following system:
\begin{subequations}\label{eq:dimensionless}
\begin{numcases}{}
			\frac{dX_i}{dt} = u_i,\label{eq:VNS1}\\
			\frac{du_i}{dt}= \eta (v(X_i,t) + a\omega_i -u_i),\label{eq:u_relaxation}\\
			d\omega_i = P_{\omega_i^{\perp}}\circ\big[ \nu \overline{\omega_i}dt   +\big(\lambda S(v)+A(v)\big)\omega_i dt\big]+ \sqrt{2D}dB_{t}^{i}, \label{eq:VNS_omega}\\
			\bar\omega_i = \frac{J_i}{|J_i|} \text{ with } J_i = \sum_{k=1}^N 		K\lp\frac{|X_i-X_{k}|}{R}\rp\omega_k, \\
			 \frac{du_i}{dt} = F_i, \label{eq:du}\\
			Re(\partial_{t}v + (v\cdot\nabla_{x}) v)  + \nabla_{x} p = \Delta_{x} v - c\sum_{i=1}^{N}F_i\delta_{X_i(t)}\nonumber\\
\qquad\qquad \qquad\qquad\quad \qquad\qquad- b  \frac{1}{N}\sum_{i=1}^N \lp\omega_i \otimes \omega_i- \frac{1}{3}\Id \rp\nabla_x\delta_{X_i(t)}, 
		\label{eq:dimensioless_NS}		\\
			\nabla_{x} \cdot v = 0,\label{eq:VNSend}
\end{numcases}
\end{subequations}
where all the variables and parameters are now dimensionless and
\begin{eqnarray*}
	Re &=& \rho_0 \frac{u_0 x_0}{\sigma} \qquad \mbox{(Reynolds number),}\\
	c &=& \frac{m u_0}{x^2_0 \sigma},\\
	b &=&  \frac{\rho_0 \beta}{x_0^2 \sigma}.
\end{eqnarray*}
Notice that the constant $\lambda$ remains unchanged with respect to the original equation; it is already a dimensionless quantity. The parameter $c$ is a measure of the particle inertia, whereas $Re$ is a measure of the fluid inertia.

\begin{remark}[Reduction to  Vicsek-Stokes coupling]
\label{rem:derivation_VS}
The Vicsek-Stokes coupling \eqref{eq:IBM1}-\eqref{eq:IBM4} is obtained from the previous system in the regime where $Re\ll 1$ as well as $c\ll 1$, $\eta\gg 1$. This corresponds to physical systems where the size (and mass) of the agents is very small (microscopic sizes). Therefore, as a simplification we can consider directly $Re=0$, $c=0$, $1/\eta=0$ thus removing the inertial terms in the Navier-Stokes equation and the force term $- c\sum_{i=1}^{N}F_i\delta_{x_i(t)}$; as well as relaxing the velocity of the particles to $u_i=v(X_i,t)+a \omega_i$. Typically the coefficient $b$ will not be small and should not be simplified.

\end{remark}

\subsubsection{The mean-field limit}

From now on, we will consider the large friction limit regime defined as follows:
\begin{definition}[Large friction limit regime]
\label{def:large_friction}
The large friction limit regime corresponds to the friction coefficient $\eta\to \infty$ in the system \eqref{eq:dimensionless} (but leaving $c$ and $Re$ to be $\mathcal{O}(1)$). Then Eq. \eqref{eq:u_relaxation} is replaced by 
$$u_i=v(X_i,t)+a\omega_i,$$ 
and the rest of equations in \eqref{eq:VNS1}-\eqref{eq:VNSend} remain unchanged.
\end{definition}

This section is devoted to proving the following:
\begin{proposition}[Mean-field limit at finite Reynolds number and finite particle inertia]
\label{prop:mean-field-NS}
Given $N$ particles, consider the following scaling of the constant $c$ in Eq. \eqref{eq:dimensioless_NS}:
\begin{equation} \label{eq:scaling_c}
c= \frac{\bar c}{N}, \quad \bar c = \mathcal{O}(1) \mbox{ as }N\to \infty.
\end{equation}
Consider also the empirical distribution associated to the dynamics of the agents in \eqref{eq:dimensionless} in the regime of large friction coefficient (Def. \ref{def:large_friction}) with the previous scaling for $c$, i.e.:
\begin{equation} \label{eq:empirical_distribution}
			f^{N}(x,\omega,t)=\frac{1}{N}\sum_{i=1}^{N}\delta_{x_i(t)}(x)\otimes\delta_{\omega_i(t)}(\omega),
		\end{equation}
		where $\delta_{x_i(t)}(x)$ and $\delta_{\omega_i(t)}(\omega)$ denote the Dirac delta at $x_i(t)$ and $\omega_i(t)$ on $\R^3$  and $\mathbb{S}^2$, respectively. Assume that  $f^N$ converges weakly to $f=f(x,\omega,t)$ as the number of agents $N\to \infty$. Then, the limit $f$ satisfies the following system:
\begin{subequations} \label{eq:mf_VNS}
\begin{numcases}{}
		\partial_t f + \nabla_x\cdot(u_{(f,v)}f) + \nabla_{\omega}\cdot\Big( \big[P_{\omega^{\perp}} \left\{\nu \overline{\omega}_{f} +\lp\lambda S(v)+A(v)\rp\omega\right\}\big]f\Big) = D\Delta_{\omega}f, \quad
			\label{eq:mean-fieldVNS1} \\
			u_{(f,v)}(x,\omega,t)=v(x,t) + a\omega, \\
				\partial_t \big[(Re+ \bar c \rho_f) v + a\bar c  j_f \big] + \nabla_x \cdot \big[ (Re + \bar c \rho_f) v\otimes v + a \bar c (v \otimes j_f+j_f\otimes v)  \big]\nonumber\\
				\qquad\qquad+\nabla_x\cdot\big[ (a^2 \bar c + b) Q_f \big]
				=-\nabla_x \lp p+\frac{a^2\bar c}{3}\rho\rp + \Delta_x v, 
				\label{eq:mean-fieldVNS_v}\\
			\nabla_x \cdot v =0,
\end{numcases}
\end{subequations}
where  the density $\rho_f$, the flux $j_f$ and the Q-tensor $Q_f$ are given by 
\begin{equation}
\rho_f:= \int_{\mathbb{S}^2}f\, d\omega, \quad j_f:= \int_{\mathbb{S}^2}\omega f\, d\omega,\quad Q_f:= \int_{\mathbb{S}^2} \lp \omega \otimes \omega- \frac{1}{3}\emph{Id} \rp  f\, d\omega, \label{eq:density_flux_q-tensor}
\end{equation}
and 
$\bar\omega_f$ is given in Eq. \eqref{eq:bar-omega-f}.
\end{proposition}

\begin{remark}
We must assume that  $c= \mathcal{O}(1/N)$ as the number of particles $N\to \infty$. This is because in  a mean-field limit interacting terms scale like $1/N$ so that their sum acting on a given particle remains finite.
\end{remark}

Proposition \ref{prop:mean-field-NS} is consequence of the following two lemmas:
\begin{lemma}
\label{lem:flux}
Consider the large friction limit regime in Def. \ref{def:large_friction}.
The density $\rho_f$ and the flux $j_f$  given in Eq. \eqref{eq:density_flux_q-tensor}
satisfy the following equations:
\begin{eqnarray} 
&&\partial_t \rho_f + \nabla_x\cdot (\rho_f v+ a j_f)=0, \label{eq:for rho}\\
&&\partial_t j_f+\nabla_x\cdot (v\otimes j_f+ a Q_f) + \frac{a}{3}  \nabla_x \rho_f \nonumber\\
&& \qquad\qquad= \int_{\mathbb{S}^2}  P_{\omega^\perp}[\nu \bar\omega_f + (\lambda S(v)+A(v))\omega]  f\, d\omega -2D j_f. \quad \qquad\label{eq:for flux2}
\end{eqnarray}
\end{lemma}

\begin{lemma}
\label{prop:mean-field-force}
Consider the large friction limit regime in Def. \ref{def:large_friction}.
Consider also the  scaling for the constant $c$ given in Eq. \eqref{eq:scaling_c}.
Then, the mean-field limit of the force term in Eq. \eqref{eq:dimensioless_NS} is
\begin{eqnarray}
-c\sum_{i=1}^N F_i \delta_{X_i(t)}&\to &-\bar c \Big[\rho_f(x,t) \left[ \partial_t v +(v\cdot \nabla_x) v\right]+ a(j_f\cdot \nabla_x) v\nonumber\\
&&\qquad + a \int_{\mathbb{S}^2} P_{\omega^\perp}\big[\nu \bar \omega_f + (\lambda S(v)+A(v)) \omega\big] f\, d\omega - 2aD\, j_f \Big],\quad \label{eq:force-first-statement}
\end{eqnarray}
as $N\to\infty$,
where $\rho_f$, $j_f$ and
$\bar\omega_f$ are given in Eqs. \eqref{eq:density_flux_q-tensor}, \eqref{eq:bar-omega-f}.  Consequently, the limit as $N\to \infty$ of Eq.  \eqref{eq:dimensioless_NS} is given by
\begin{eqnarray}
&&\partial_t \big[(Re+ \bar c \rho_f) v + a \bar c  j_f \big] + \nabla_x \cdot \big[ (Re + \bar c \rho_f) v\otimes v + a \bar c (v \otimes j_f+j_f\otimes v) + (a^2 \bar c + b) Q_f \big]\nonumber\\
&&\qquad\qquad=- \nabla_x \lp p +\frac{a^2\bar c}{3}\rho_f\rp+ \Delta_x v. \label{eq:kinetic_v}
\end{eqnarray}
\end{lemma}

The proof of these two Lemmas is given at the end of this section. We prove first Prop. \ref{prop:mean-field-NS}:
\begin{proof}[Proof of Prop. \ref{prop:mean-field-NS}]
The mean-field limit equation for the density $f$ is computed analogously as in  Prop. \ref{prop:mean-field}. We just need to compute the mean-field limit equation for the fluid velocity $v$ in Eq. \eqref{eq:dimensioless_NS} and this is done in Lem. \ref{prop:mean-field-force}.
\end{proof}

\begin{proof}[Proof of Lem. \ref{lem:flux}]
As in Prop. \ref{prop:mean-field},  we have that the density $f$ satisfies Eq. \eqref{eq:mean-fieldVNS1}. 
To obtain Eq. \eqref{eq:for rho} for $\rho_f$ we integrate this equation with respect to $\omega$.

To obtain Eq. \eqref{eq:for flux2} for the flux $j_f$ we multiply the kinetic equation \eqref{eq:mean-fieldVNS1} by $\omega$ and integrate over $\omega$:
\begin{eqnarray}
&&\partial_t j_f + \nabla_x \cdot (v\otimes j_f +a Q_f)+ \frac{a}{3}\nabla_x \rho_f\nonumber\\
&&\qquad\qquad+\int_{\mathbb{S}^2} \omega \nabla_\omega \cdot \Big[ P_{\omega^\perp}[\nu \bar\omega_f + (\lambda S(v)+A(v))\omega] \Big] f\, d\omega = D \int_{\mathbb{S}^2}\omega \Delta_\omega f \, d\omega,\qquad \label{eq: aux_j}
\end{eqnarray}
where we used that $(v\cdot \nabla_x) j_f= \nabla_x \cdot (v\otimes j_f)-(\nabla_x \cdot v) j_f$ and $\nabla_x\cdot v=0$.
Next, we recast the last two terms of this equation. Firstly, it holds that
\begin{equation} \label{eq:aux_laplacian}
D\int_{\mathbb{S}^2}\omega\Delta_\omega f\,d\omega= -2D\int_{\mathbb{S}^2}\omega f\, d\omega = -2Dj_f,
\end{equation}
using integration by parts and the fact that the laplacian in the sphere satisfies $\Delta_\omega (\omega \cdot u) = -2 (\omega\cdot u)$ for any vector $u\in\R^3$ (this is the spherical harmonic of degree 1 in $\mathbb{S}^2$, see \cite{frouvelle2012continuum}). Secondly, it holds that
$$
\int_{\mathbb{S}^2} \omega \nabla_\omega \cdot \Big[ P_{\omega^\perp}[\nu \bar\omega_f + (\lambda S(v)+A(v))\omega] \Big] f\, d\omega=- \int_{\mathbb{S}^2} P_{\omega^\perp}[\nu \bar\omega_f + (\lambda S(v)+A(v))\omega] f\, d\omega.
$$
A proof of the last equality can be found in Prop. \ref{prop:integral_divergence}. Substituting this last expression and Eq. \eqref{eq:aux_laplacian} into Eq. \eqref{eq: aux_j} we conclude Eq. \eqref{eq:for flux2} for $j_f$.
\end{proof}

\begin{proof}[Proof of Lemma \ref{prop:mean-field-force}]
We consider the following decomposition:
 $$-c\sum_{i=1}^N F_i \delta_{X_i(t)} dt= -\frac{\bar c}{N}\sum_{i=1}^N 
du_i \delta_{X_i(t)}= T_1^N+T_2^N,$$
where
\begin{eqnarray*}
T_1^N&=&-\frac{\bar c}{N}\sum^N_{i=1} dv(X_i(t),t)\,  \delta_{X_i(t)},\\
T_2^N &=& -\frac{\bar c}{N}\sum^N_{i=1}a\, d\omega_i(t) \delta_{X_i(t)}.
\end{eqnarray*}

For  the limit of  $T_1^N$ as $N\to\infty$, we have,  using \eqref{eq:dimensioless_NS} and ignoring the Dirac deltas (in Newtonian mechanics, self-forces are ignored to keep  the expressions finite) that
\begin{eqnarray*}
T_1^N &=& -\frac{\bar c}{N}dt\sum^N_{i=1}\left[ \partial_t v +(v\cdot \nabla_x) v+a (\omega_i(t)\cdot\nabla_x)v\right] (X_i(t),t) \,  \delta_{X_i(t)}\\
&=&  -\frac{\bar c}{N}dt\sum^N_{i=1}\int_{\mathbb{S}^2}\left[ \partial_t v +(v\cdot \nabla_x) v+a (\omega\cdot\nabla_x)v\right] (x,t) \,  \delta_{X_i(t)}\delta_{\omega_i(t)}\, d\omega\\
&=&-\bar c\,dt\int_{\mathbb{S}^2} \left[ \partial_t v +(v\cdot \nabla_x) v\right](x,t)\,  \lp \frac{1}{N}\sum^N_{i=1}\delta_{X_i(t)} \delta_{\omega_i(t)}\rp\, d\omega\\
&& -a\bar c\, dt\int_{\mathbb{S}^2}(\omega\cdot\nabla_x)v\lp \frac{1}{N}\sum^N_{i=1}\delta_{X_i(t)} \delta_{\omega_i(t)}\rp\, d\omega \\
&=&-\bar c\, dt\left[ \partial_t v +(v\cdot \nabla_x) v\right](x,t)\int_{\mathbb{S}^2} f^N(x,\omega,t)\, d\omega\\
&& -a \bar c\, dt\left[\lp\int_{\mathbb{S}^2}\omega f^N\, d\omega \rp \cdot \nabla_x \right] v\\
&\to &-\bar c\, dt \left[\rho_f(x,t) \left[ \partial_t v +(v\cdot \nabla_x) v\right]+a (j_f\cdot \nabla_x) v \right], \quad \mbox{as }N\to \infty.
\end{eqnarray*}

To compute the limit of $T^N_2$ we first recast the stochastic differential equation \eqref{eq:VNS_omega} for $\omega_i$, which is expressed in Stratonovich sense, in its equivalent It\^o's form (see \cite[Th. (30.14) p. 185]{rogers2000diffusions}, also  \cite{bolley2012mean}):
$$d\omega_i = P_{\omega_i^{\perp}} (\nu \overline{\omega_i}dt   +\big(\lambda S(v)+A(v)\big)\omega_i dt)+ \sqrt{2D}dB_{t}^{i}- 2D\omega_i\, dt.$$
With this, we consider the decomposition of $T^N_2$ into
\begin{eqnarray*}
T_2^N &=& -\frac{\bar c}{N}dt\sum^N_{i=1}a\, P_{\omega_i^\perp}\lp \nu \bar \omega_i + (\lambda S(v)+A(v)) \omega_i \rp \delta_{X_i(t)} \\
&&-\frac{\bar c}{N}\sum^N_{i=1}a\, P_{\omega_i^\perp}\lp \sqrt{2D} dB^i_t\rp\delta_{X_i(t)}\\
&&  + \frac{a\bar c}{N}2 D\omega_i\delta_{X_i(t)}\, dt\\
&=:& T_{21}^N+T_{22}^N+T_{23}^N.
\end{eqnarray*}

To compute the limit of $T_{21}^N$ we define
$$g^N(x,\omega,v) = \bar c P_{\omega^\perp}\lp \nu \bar \omega^N + (\lambda S(v)+A(v)) \omega \rp\,dt, $$
where 
$$\bar \omega^N (x)= \frac{J}{|J|},\mbox{ with } J(x)= \sum_{k=1}^N K\lp\frac{|x-X_k|}{R} \rp\omega_k.$$
With these notations we rewrite:
\begin{eqnarray*}
T_{21}^N &=& -\frac{1}{N}\sum_{i=1}^N a g^N(x,\omega_i(t),v(t)) \delta_{X_i(t)}\\
&=&-\frac{1}{N}\sum_{i=1}^N\int_{\mathbb{S}^2} a g^N(x,\omega,v(t)) \delta_{X_i(t)}\delta_{\omega_i(t)} \,d\omega\\
&=& -\int_{\mathbb{S}^2} a g^N(x,\omega,v(t)) f^N(x,\omega,t) d\omega\\
&\to &-\int_{\mathbb{S}^2} a g(x,\omega,v(t)) f(x,\omega,t) d\omega, \quad \mbox{ as }N\to \infty,
\end{eqnarray*}
where
$$ g(x,\omega,v) = \bar cP_{\omega^\perp}\lp \nu \bar \omega_f + (\lambda S(v)+A(v)) \omega \rp\, dt, $$
and where $\bar \omega_f$ is given in Eq. \eqref{eq:bar-omega-f}. This leads to
$$T_{21}^N\to -a\bar c \int_{\mathbb{S}^2}P_{\omega^\perp}\lp \nu \bar \omega_f + (\lambda S(v)+A(v)) \omega \rp\, f\, d\omega\,  dt. $$

For the term $T_{22}^N$ we have that
\begin{eqnarray*}
T_{22}^N&=&-\frac{\bar c}{N}\sum^N_{i=1}a\, P_{\omega_i^\perp}\lp \sqrt{2D} dB^i_t\rp\delta_{X_i(t)}\\
&=&-a \bar c\sqrt{2D}\int_{\mathbb{S}^2}\frac{1}{N}\sum_{i=1}^N \lp P_{\omega_i^\perp} dB^i_t\, \delta_{X_i(t)}\delta_{\omega_i(t)} \rp\, d\omega.
\end{eqnarray*}
For any test function $\varphi=\varphi(x,\omega)$ we have that
$$\langle P_{\omega_i^\perp} dB^i_t\, \delta_{X_i(t)}\delta_{\omega_i(t)} , \varphi \rangle = \varphi(X_i(t), \omega_i(t)) P_{\omega_i^\perp} dB^i_t,$$
where $\langle\cdot,\cdot \rangle$ denotes the duality brackets.
Now it holds that
\begin{equation} \label{eq:BM_projection}
\frac{1}{N}\sum_{i=1}^N\varphi(X_i(t), \omega_i(t)) P_{\omega_i^\perp} dB^i_t = \frac{1}{N}\sum_{i=1}^N\varphi(X_i(t), \omega_i(t) )\lp  dB^i_t- (\omega_i(t) \cdot dB^i_t ) \omega_i(t) \rp.
\end{equation}
The term  $dB_t=B_{t+dt}-B_t$ denotes Brownian motion increments, by the properties of Brownian motion, we have that $dB_t$ is normally distributed with mean 0 and variance $dt$, i.e., $dB_t \sim \mathcal{N}(0, dt)$. 
For fixed $t$, the following term is a gaussian random variable
$$\frac{1}{N}\sum_{i=1}^N\varphi(X_i(t), \omega_i(t) ) dB^i_t,$$
since it is the sum of independent gaussian random variables (notice that for fixed $t$, $\omega_i(t)$ takes a particular fixed value and it is not random). Particularly, 
its expectation $\mathbb{E}$ is zero:
$$\mathbb{E}\lp\frac{1}{N}\sum_{i=1}^{N}\varphi(X_i(t), \omega_i(t) )  dB^i_t \rp =\frac{1}{N} \mathbb{E}\lp dB^1_t \rp\sum_{i=1}^{N}\varphi(X_i(t), \omega_i(t) )=0,$$
(since $\mathbb{E}(dB^1_t)=0$) and, moreover, since the Brownian motions are independent (and hence $\mathbb{E}(dB^i_tdB^j_t)=0$ if $i\neq j$), it holds that the variance is zero too in the limit $N\to \infty$:
\begin{eqnarray*}
\mbox{Var}\lp\frac{1}{N}\sum_{i=1}^{N} \varphi(X_i(t), \omega_i(t) ) dB^i_t \rp &=& \mathbb{E}\left[\lp\frac{1}{N}\sum_{i=1}^{N} \varphi(X_i(t), \omega_i(t) )  dB^i_t \rp^2\right] \\
&&- \mathbb{E}^2 \lp\frac{1}{N}\sum_{i=1}^{N} \varphi(X_i(t), \omega_i(t) )  dB^i_t \rp\\
&=& \frac{1}{N^2}\mathbb{E}(dB^1_t)^2 \sum_{i=1}^{N}[\varphi(X_i(t), \omega_i(t) ) ]^2\\
&=& \frac{dt}{N^2}\sum_{i=1}^{N}[\varphi(X_i(t), \omega_i(t) ) ]^2\to 0 \mbox{ as }N\to\infty,
\end{eqnarray*}
where we used that $\mathbb{E}(dB^1_t)^2=dt$ and the fact that
$$\frac{1}{N}\sum_{i=1}^N [\varphi(X_i(t),\omega_i(t))]^2= \langle f^N, \varphi^2\rangle \to \langle f, \varphi^2\rangle <\infty, \quad \mbox{as }N\to\infty.$$
One can show analogously that the term $(\omega_i\cdot dB_t^i)\omega_i$ in Eq. \eqref{eq:BM_projection} satisfies the same properties since each component of $B^i_t$ is also a Brownian motion (in 1-dimension).
From this we conclude,  that
$$T_{22}^N(x,\omega,t)\to 0, \quad \mbox{ as } N\to \infty.$$

Finally, one can see following  analogous computations to the previous ones that
$$T_{23}^N \to 2a \bar cD\int_{\mathbb{S}^2} \omega f\, d\omega\, dt= 2a \bar cD j_f\, dt.$$
Putting all the terms together we conclude the proof of statement \eqref{eq:force-first-statement}.

\medskip
We prove next Eq. \eqref{eq:kinetic_v}. Using Eq. \eqref{eq:force-first-statement}, the mean-field limit for  the fluid velocity $v$ \eqref{eq:dimensioless_NS} corresponds to:
\begin{eqnarray*}
Re (\partial_t v + (v\cdot \nabla_x) v) &=& - \nabla_x p + \Delta_x v\\
&& -\bar c \rho_f [\partial_t v + (v\cdot \nabla_x) v] - a\bar c (j_f \cdot \nabla_x) v\\
&&  -a\bar c \int_{\mathbb{S}^2} P_{\omega^\perp}\big[ \nu \bar \omega_f + \big( \lambda S(v)+ A(v)\big) \omega\big] f \, d\omega + 2a\bar c D \, j_f\\
&& -b \nabla_x \cdot Q_f.
\end{eqnarray*}
Now, using that $\nabla_x\cdot v=0$, as well as $\nabla_x \cdot (v\otimes v)= (v \cdot \nabla_x)v$ and the equation for the density $\rho_f$ in Eq. \eqref{eq:for rho}, the previous expression is recast into
\begin{eqnarray*}
&&\hspace*{-2cm}\partial_t \big[(Re +\bar c \rho_f) v\big] +\nabla_x \cdot \big[ (Re + \bar c \rho_f) v\otimes v\big] \\
&=& - \nabla_x p + \Delta_x v\\
&&-a\bar c (\nabla_x\cdot j_f)v -  a\bar c (j_f \cdot \nabla_x) v\\
&&-a \bar c \int_{\mathbb{S}^2} P_{\omega^\perp}\big[ \nu \bar \omega_f + \big( \lambda S(v)+ A(v)\big) \omega\big] f \, d\omega + 2a\bar c D \, j_f\\
&& -b \nabla_x \cdot Q_f.
\end{eqnarray*}
Finally, from this expression we obtain Eq. \eqref{eq:kinetic_v} using Eq. \eqref{eq:for flux2} for the flux $j_f$  and the fact that $(\nabla_x \cdot j_f) v+ (j_f \cdot \nabla_x )v = \nabla_x \cdot (j_f\otimes v)$.

\end{proof}

\begin{remark}
Notice that Eq. \eqref{eq:mean-fieldVNS_v} for  the velocity of the fluid $v$  is in conservative form. From it, assuming that the domain has no boundaries and the solution vanishes at large distances, we conclude that 
$$\partial_t \int_{\R^3} (Re\, v + \bar c \rho_f v+ a \bar c j_f)\, dx=0,$$
and therefore, the total momentum of the system is conserved, as expected, given the conservation of the total momentum in the individual based model.
\end{remark}

\subsubsection{Macroscopic equations}

To obtain the macroscopic equations, we scale the mean-field limit system from Prop. \ref{prop:mean-field-NS} analogously as done in Sec. \ref{sec:mean_field}:
\begin{equation}
			\begin{cases}
			&\eps\big[\partial_t f^\eps + \nabla_x\cdot(u_{(f^\eps,v^\eps)}f^\eps) \big]+ \nabla_{\omega}\cdot\Big( \big[P_{\omega^{\perp}} \left\{\nu \overline{\omega}_{f^\eps} +\eps \lp\lambda S(v^\eps)+A(v^\eps)\rp\omega\right\}\big]f^\eps\Big) = D\Delta_{\omega}f^\eps, \\
			&u_{(f^\eps,v^\eps)}(x,\omega,t)=v^\eps(x,t) + a\omega, \label{eq:VNS-rescaled}\\
			& \bar \omega^\eps_{f}= \frac{J^\eps_{f}}{|J^\eps_{f}|}, \quad J^\eps_{f} = \int_{\mathbb{S}^2 \times \R^3} \omega K\lp\frac{|x-y|}{\sqrt{\eps}R} \rp f\, d\omega dy, \\
			&	\partial_t \big[(Re+ \bar c \rho_{f^\eps}) v^\eps + \bar c a j_{f^\eps} \big] + \nabla_x \cdot \big[ (Re + \bar c \rho_{f^\eps}) v^\eps\otimes v^\eps + a \bar c (v^\eps \otimes j_{f^\eps}+j_{f^\eps}\otimes v^\eps)  \big]\\
				&\qquad\qquad+\nabla_x\cdot\big[ (a^2 \bar c + b) Q_{f^\eps} \big]
				=-\nabla_x \lp p^\eps +\frac{a^2\bar c}{3}\rho_{f^\eps}\rp+ \Delta_x v^\eps,  \\
			&\nabla_x \cdot v^\eps =0.
			\end{cases} 
		\end{equation}
		
Finally, we conclude the

\begin{theorem}[Macroscopic equations at high Reynolds number]
\label{cor:macro-VNS}
Consider the scaled system \eqref{eq:VNS-rescaled}. When $\eps \to 0$, it holds (formally) that
$$(f^\eps, v^\eps, p^\eps)\to ( f = \rho M_\Omega, v, p),$$
where $\rho=\rho(x,t)\geq 0$ and $\Omega = \Omega (x,t)\in \mathbb{S}^{2}$ are the limits of the local density $\rho^\eps$ and 
the local mean orientation $\Omega_{f^\eps}$ in Eqs. \eqref{eq:density_local},\eqref{localaverageorient}, respectively.
Moreover, if the convergence is strong enough and $\Omega$, $\rho$, $v$ and $p$ are smooth enough, they satisfy the following coupled system:
\begin{subequations}\label{eq:macro_VNS}
\begin{numcases}{}
		 \partial_t\rho + \nabla_x\cdot(\rho U)  = 0, \label{eq:continuity_equationVNS} \\
		 \rho\partial_{t}\Omega + \rho (V\cdot\nabla_x)\Omega +\frac{a}{\kappa} P_{\Omega^{\perp}}\nabla_x\rho= \gamma P_{\Omega^{\perp}}\Delta_x(\rho\Omega) +\rho P_{\Omega^{\perp}} \lp\tilde\lambda S(v) +A(v)\rp \Omega , \label{eq:hydro2VNS} \\
			\partial_t \big[(Re+ \bar c \rho) v + c_1a\bar c  \rho \Omega \big] + \nabla_x \cdot \big[ (Re + \bar c \rho) v\otimes v + c_1a \bar c \rho (v \otimes \Omega+\Omega\otimes v)  \big]\nonumber\\
				\qquad\qquad+\nabla_x\cdot\big[ (a^2 \bar c + b) \mathcal{Q} \big]
				=-\nabla_x\tilde p + \Delta_x v, \label{eq:hydro3VNS}\\ 
		\nabla_x\cdot v=0 \label{eq:hydro4VNS}, 
	\end{numcases}
\end{subequations}
	where 
	\begin{eqnarray*}
		&& U = a c_1 \Omega +v, \quad
		V = a c_2 \Omega + v, \quad
		\mathcal{Q} = c_4\lp \Omega \otimes \Omega-\frac{1}{3}\emph{Id} \rp, \quad
		\tilde p = p+\frac{a^2\bar c}{3}\rho,
	\end{eqnarray*}
	and where the constants $c_1,\hdots, c_4$, $k_0$, $\tilde \lambda$ and $\gamma$ are given by Eqs. \eqref{eq:c1_th}--\eqref{eq:c4}, \eqref{eq:def_k0}, \eqref{eq:lambda0}, respectively; and $\kappa=\nu/D$.

\end{theorem}
The proof of this result is direct from the one of Th. \ref{th:hydro_limit} since most of the terms are computed there. For the extra terms that depend on $j_{f^\eps}$ one just needs to remember that $j_{f^\eps} \to c_1 \rho \Omega$ as $\eps \to 0$.

\begin{remark}[Discussion of the results in Th. \ref{cor:macro-VNS}.] 
\label{rem:discussion_SOH-NV}
 Notice firstly that when $\bar c=0$ and $Re=0$, we recover the SOH-Stokes system \eqref{eq:macro_SOH_Stokes} as expected, since in that case the individual based model corresponds to the Vicsek-Stokes coupling \eqref{eq:Vicsek_Stokes}), see Rem \ref{rem:derivation_VS}. The interpretations of the equations for $\rho$ and $\Omega$ are the same  as for the SOH-Stokes, since the equations are the same. The difference with respect to the SOH-Stokes system is Eq. \eqref{eq:hydro3VNS}. This equation gives the evolution over time of the total momentum of the fluid and the particles corresponding to:
$$(Re + \bar c \rho) v + c_1 a \bar c \rho \Omega.$$
The second term in  \eqref{eq:hydro3VNS} corresponds to the momentum flux and it is divided in two contributions. Firstly,
$$ (Re+ \rho) v\otimes v$$
corresponds to the momentum flux generated by the fluid and by the passive transport of the particles by the fluid. Secondly, the term corresponding to
$$c_1 a \bar c  \rho (v \otimes \Omega + \Omega\otimes v)$$
gives the momentum flux through the exchange between fluid velocity $v$ and particles velocity $c_1 a \Omega$. Notice that the momentum flux is given by a symmetric matrix.  The term $(a^2 \bar c ~+~b )\mathcal{Q}$ gives an extra-stress tensor coming from the active nature of the particles and splits into a contribution coming from the dipolar force exerted by the particles (corresponding to the contribution given by the constant $b$), on the one hand, and from their net force (corresponding to the contribution given by the product $a^2\bar c$), on the other hand.
\end{remark}

\section{Conclusions}
 \label{sec:conclusion} 
  
In this paper we have presented the macroscopic derivation of a coupled Vicsek-Stokes system. This coupling describes collective motion in a fluid in a low Reynolds number regime. The fluid is described by Stokes system and the collective motion by the Vicsek model, which represents phenomenologically the interactions between neighbouring agents mediated by the fluid.  The coupling is obtained by taking into account the interactions between the agents and the fluid. This involves, particularly, Jeffery's equation that expresses the influence of a viscous fluid on spheroidal particles, on the one hand, and the force exerted by the agents on the fluid due to the dipolar force created by their self-propulsion motion, on the other hand.  

The coarse-grained model corresponds to a  Self-Organised Hydrodynamics and Stokes coupling. Interestingly, we have shown that Jeffery's equation is coarse-grained into Jeffery's equation but with a different value for the shape parameter. The  linear stability analysis shows that both pullers and pushers have unstable modes, but the instability of pullers disappears in the case of rod-like particles.

At the end, we have extended the Vicsek-Stokes coupling  into two directions: firstly, we take into account volume exclusion to avoid concentration effects in the dynamics; secondly,  we consider a finite Reynolds number and finite particle inertia regime to model systems where the particles' mass and size is large such as fish.

Finally, these results open many exciting paths to be explored, for example, one could consider the coupling of the Vicsek model with other types of fluid dynamics (given e.g. by Darcy's law, Brinkmann law, non-Newtonian fluids). Also, it would be interesting to  perform numerical simulations of the dynamics to confirm the stability analysis and to apply these models to the investigation of real-life systems like sperm and bacterial suspensions.

\bigskip

{\bf Acknowledgements:}
PD  acknowledges  support  by  the  Engineering  and  Physical  Sciences  Research  Council
(EPSRC) under grants no.  EP/M006883/1 and EP/P013651/1, by the Royal Society and the Wolfson Foundation through a Royal Society Wolfson Research Merit Award no.  WM130048 and by the National Science Foundation (NSF) under grant no.  RNMS11-07444 (KI-Net).  PD is
on leave from CNRS, Institut de Math\'ematiques de Toulouse, France.\\
S.M.A. was supported by the British Engineering and Physical Research Council under grant ref: EP/M006883/1.\\
HY acknowledges the support by Division of Mathematical Sciences [KI-Net NSF RNMS grant number 1107444]; DFG Cluster of Excellence
Production technologies for high-wage countries [grant number DFG STE2063/1-1], [grant
number HE5386/13,14,15-1]. HY and FV gratefully acknowledges the hospitality of the
Department of Mathematics, Imperial College London, where part of this research was
conducted.

\section*{Data statement}

No new data was generated in the course of this research.

\appendix

\section{Some proofs and properties}

\begin{proof}[Proof of Prop. \ref{prop:conservation_momentum}]
The total momentum is given by $\int\rho_0 v(x,t) dx + \sum^N_{i=1} m_i u_i$. It is a direct computation to check that its derivative is zero. The total angular momentum for the system is given by:
\begin{equation}
\int{x \times (\rho_{0}\,v(x,t))\, dx} + \sum_{i=1}^{N}(X_i\times F_i).
\end{equation}

We have that
\begin{align*}
\frac{d}{dt}\lp\int{\rho_{0}x\times v\,dx}\rp &= \int{\rho_{0}x\times \partial_{t}v\,dx}\\
&= - \int{\rho_{0}x\times\nabla_x\cdot(v\otimes v)\,dx} -  \int{x\times \nabla_x p\,dx} + \int{\sigma x\times \Delta v\,dx}\\
&~~~ - \sum_{i=1}^{N}\int{x\times F_i\delta_{X_i}\,dx} - \sum_{i=1}^{N}\int{x\times \lp\omega_i\otimes \omega_i -\frac{1}{3}\Id \rp\nabla_x\delta_{X_i}\,dx}\\
&=: I_1+I_2+I_3+I_4+I_5.
\end{align*}
One can check directly with the help of the L\'evy-Civita symbol to compute the vector products (and integration by parts in some cases) that
\begin{eqnarray*}
&&I_1=I_2=I_3=I_5=0,\\
&& I_4= -\sum^N_{i=1} X_i \times F_i.
\end{eqnarray*}
 Notice, that $I_5=0$ thanks to  $(\omega_i\otimes\omega_i-\Id/3)$ being a symmetric matrix. Therefore, the only term that does not vanish is $I_4$ and it is compensated by the angular momentum of the agents.
\end{proof}

\begin{proposition} \label{prop:integral_divergence}
For any vector $u\in R^3$, it holds
$$\int_{\mathbb{S}^2}\omega \nabla_\omega \cdot(P_{\omega^\perp}u)\, d\omega = -\int_{\mathbb{S}^2}P_{\omega^\perp}u \, d\omega.$$
\end{proposition}
\begin{proof}
This can be proven as follows: for any vector $q\in\R^3$
\begin{eqnarray*}
q\cdot \int_{\mathbb{S}^2}\omega \nabla_\omega \cdot(P_{\omega^\perp}u)\, d\omega &=& \int_{\mathbb{S}^2}(q\cdot\omega) \nabla_\omega \cdot(P_{\omega^\perp}u)\, d\omega\\
&=& -\int_{\mathbb{S}^2}\nabla_\omega(q\cdot \omega) \cdot (P_{\omega^\perp}u)\, d\omega \\
&=&-\int_{\mathbb{S}^2}P_{\omega^\perp}q \cdot (P_{\omega^\perp}u)\, d\omega \\
&=& -q \cdot \int_{\mathbb{S}^2} (P_{\omega^\perp}u)\, d\omega
\end{eqnarray*}
given that $\nabla_\omega (\omega\cdot q)= P_{\omega^\perp}q$ and $P_{\omega^\perp}q \cdot P_{\omega^\perp} u= q \cdot P_{\omega^\perp}u$ for any pair of vectors $q,u\in \R^3$. From this we conclude the result.
\end{proof}

\bibliographystyle{abbrv}
\bibliography{biblio}
\end{document}